\documentclass[11pt]{article}
\usepackage[margin=1in]{geometry}
\usepackage{times}
\usepackage[numbers]{natbib}
\usepackage{subfigure}
\usepackage{graphicx,amssymb,amsmath}
\usepackage{epsfig}
\usepackage{algorithm}
\usepackage{latexsym}
\usepackage{url}
\newtheorem{theorem}{Theorem}[section]
\newtheorem{lemma}[theorem]{Lemma}

\newtheorem{corollary}[theorem]{Corollary}
\newtheorem{definition}[theorem]{Definition}

\newcommand{\sq}{\hbox{\rlap{$\sqcap$}$\sqcup$}}
\newcommand{\qed}{\hspace*{\fill}\sq}
\newenvironment{proof}{\noindent {\bf Proof.}\ }{\qed\par\vskip 4mm\par}

\begin{document}
\setcitestyle{aysep={},authoryear}
\title{\bf Optimal mechanisms with simple menus }

\author{Zihe Wang \\
   IIIS, Tsinghua University \\
   wangzihe11@mails.tsinghua.edu.cn
   \and
   Pingzhong Tang \\
   IIIS, Tsinghua University \\
   kenshin@tsinghua.edu.cn \\
   }
\maketitle 
\begin{abstract}
We consider optimal mechanism design for the case with one buyer and two items. The buyer's valuations towards the two items are independent and additive. In this setting, optimal mechanism is unknown for general valuation distributions. We obtain two categories of structural results that shed light on the optimal mechanisms. These results can be summarized into one conclusion: under certain conditions, the optimal mechanisms have simple menus.

The first category of results state that, under certain mild condition, the optimal mechanism has a monotone menu. In other words, in the menu that represents the optimal mechanism, as payment increases, the allocation probabilities for both items increase simultaneously. This theorem complements Hart and Reny's recent result regarding the nonmonotonicity of menu and revenue in multi-item settings. Applying this theorem, we derive a version of revenue monotonicity theorem that states stochastically superior distributions yield more revenue. Moreover, our theorem subsumes a previous result regarding sufficient conditions under which bundling is optimal~\citep{Hart2012a}. The second category of results state that, under certain conditions, the optimal mechanisms have few menu items. Our first result in this category says, for certain distributions, the optimal menu contains at most 4 items. The condition admits power (including uniform) density functions. Our second result in this category works for a weaker condition, under which the optimal menu contains at most 6 items. This condition includes exponential density functions. Our last result in this category works for unit-demand setting. It states that, for uniform distributions, the optimal menu contains at most 5 items. All these results are in sharp contrast to Hart and Nisan's recent result that finite-sized menu cannot guarantee any positive fraction of optimal revenue for correlated valuation distributions.
\end{abstract}
\vspace{-0.2in}
\section{Introduction}

Optimal mechanism design has been a topic of intensive research over the past thirty years. The general problem is, for a seller, to design a revenue-maximizing mechanism for selling $k$ items to $n$ buyers, given the buyers' valuations distributions but not the actual values. A special case of the problem, where there is only one item ($k=1$) and buyers have independent valuation distributions towards the item, has been resolved by Myerson's seminal work~\citep{Myerson81}. Myerson's approach has turned out to be quite general and has been successfully applied to a number of similar settings, such as ~\citep{Maskin89:Optimal,Jehiel1996,Levin97,Ledyard2007,Deng2011}, just to name a few.

While this line of work has flourished, it does not deepen our understanding of the cases with more than one items ($k>1$). In fact, even for the simplest multi-item case, where there are two independent items (k=2) and one buyer (n=1) with additive valuations, a direct characterization of the optimal mechanism is still open for general, especially continuous, valuation distributions.

When the distributions are discrete, ~\citet{Daskalakis2011,Cai2012a,Cai2012b} show that the general optimal mechanism ($k>1$) is the solution of a linear program. They provide different methods to solve the linear program efficiently. For continuous distributions, ~\citet{Chawla2010,Cai13} study the possibility of using simple auctions to approximate optimal auctions. In addition, ~\citet{Daskalakis12,Cai13} provide PTAS of the optimal auction under various assumptions on distributions.

Zoom in and look at the case with two independent items and a single buyer, much progress has been made in this particular setting.~\citet{Hart2012a} investigate two simplest forms of auctions: selling the two items separately and selling them as a bundle. They prove that selling separately obtains at least one half of the optimal revenue while bundling always returns at least one half of separate sale revenue. They further extend these results to the general case with $k$ independent items: separate sale guarantees at least a $\frac{c}{log^2k}$ fraction of the optimal revenue; for identically distributed items, bundling guarantees at least a $\frac{c}{log~k}$ fraction of the optimal revenue.~\citet{Yao2013} tighten these lower bounds to $\frac{c}{log~k}$ and $c$ respectively.
Under some technical assumptions, ~\cite{Daskalakis13} show close relation between mechanism design and transport problem and use techniques there to solve for optimal mechanisms in a few special distributions.~\citet{Hart2013} investigate how the ``menu size'' of an auction can effect the revenue and show that revenue of any finite menu-sized auction can be arbitrarily far from optimal (thus confirm an earlier consensus that restricting attention to deterministic auctions, which has an exponentially-sized (at most) menu, indeed loses generality). On the economic front,~\citet{Manelli2006,Manelli2007,Pavlov2011a,Pavlov2011b} obtain the optimal mechanisms in several specific distributions (such as both items are distributed according to uniform [0,1]). We will discuss these results in much more detail as we proceed to relevant sections.

In this paper, we study the case with one buyer and two independent items, in hopes of a direct characterization of \emph{exact} optimal mechanisms. We obtain several exciting structural results. \emph{Our overall conclusion is that, under fairly reasonable conditions, optimal mechanisms has ``simple'' menus.} We summarize our results below into two parts, based on the conditions under which the results hold, as well as different interpretations of simplicity.

For ease of presentation, we need the following definition: for a density function $h$, the \emph{power rate} of $h$ is $PR(h(x))=\frac{xh'(x)}{h(x)}$.
\begin{itemize}
\item Part I (Section~\ref{sec:part1}). If density functions $f_1$ and $f_2$ satisfy $PR(f_1(x))+PR(f_2(y))\leq -3,~\forall x,y$. The optimal mechanism has a monotone menu -- sort the menu items in ascending order of payments, the allocation probabilities of both items increase simultaneously -- a desirable property that fails to hold in general (cf~\citep{Hart2012b}). Our result complements this observation in general and has two important implications in particular.
    \vspace{-0.1in}
    \begin{enumerate}
    \item ~\cite[Theorem 28]{Hart2012a}. Hart and Nisan show that, if two item distribution are further identical (i.e., $f_1=f_2$), bundling sale is optimal. Our result subsumes this theorem as a corollary.
    \vspace{-0.1in}
    \item A revenue monotonicity theorem. Based on menu monotonicity theorem, we are able to prove that, stochastically superior distributions yield higher revenue, another desirable property that fails to hold in general.
    \end{enumerate}
    \vspace{-0.1in}
    Our proof is semi-constructive in the sense that we fix part of the buyer utility function (for this part, relation between revenue and buyer utility  is unknown/undesirable) and construct the remainder of the utility function (for this part, relation between revenue and buyer utility is known/desirable). This technique might be of independent interest.
\item Part II. (Section~\ref{sec:part2}). If the density functions $f_1$ and $f_2$ satisfy $PR(f_1(x))+PR(f_2(y))\geq -3,~\forall x,y$. The optimal mechanisms often contain few menu items. In particular,
    \vspace{-0.1in}
    \begin{enumerate}
    \item If both $PR(f_1(x))$ and $PR(f_2(y))$ are constants, the optimal mechanism contains at most 4 menu items. The result is tight. Constant power rate is satisfied by a few interesting classes of density functions, including power functions $h(x)=ax^b$ and uniform density as a special case. This is consistent with earlier results for uniform distributions~\citep{Manelli2006,Pavlov2011a}: the optimal mechanisms indeed contain four menu items.
    \vspace{-0.05in}
    \item If $PR(f_1(x))\leq y_Af_2(y_A)-2\textrm{ and }PR(f_2(y))\leq x_Af_1(x_A)-2,$ where $x_A$ and $y_A$ are the lowest possible valuations for item one and two respectively, the optimal mechanism contains 2 menu items. 
    \vspace{-0.1in}
    \item If we relax the condition to be that both $PR(f_1(x))$ and $PR(f_2(y))$ are monotone, then the optimal mechanism contains at most 6 menu items. This condition is general and admits density functions such as exponential density and polynomial density.
    \vspace{-0.1in}
    \item Our last result requires the buyer demands at most one item. Under this condition, for uniform densities, the optimal mechanism contains at most 5 menu items. The result is tight.
    \end{enumerate}
    \vspace{-0.1in}
    These results are in sharp contrast to Hart and Nisan's recent result that there is some distribution where finite number of menu items cannot guarantee any fraction of revenue~\citep{Hart2013}. Here we show that, for several wide classes of distributions, the optimal mechanisms have a finite and even extremely simple menus.

    Our proofs for this part are based on Pavlov's characterization and careful analyses of how the revenue changes as a function of the buyer's utility. A rough line of reasoning is as follows, the ``extreme points'' in the set of convex utility functions on the boundary values are piecewise linear functions. Since the utility on the boundaries contains only few linear pieces and the utility on inner values are linearly related to that on the boundary, it must be the case that the utility function on the inner points contains only few linear pieces as well. In other words, the mechanism only contains few menu items.
\end{itemize}

Our results not only offer original insights of ``what do optimal mechanisms look like'', but are also in line with the ``simple versus optimal'' literature (cf~\citep{Hartline09,Hart2012a}): in our case, simple mechanisms are exactly optimal.
\vspace{-0.1in}
\section{The setting}

We consider a setting with one seller who has two distinct items for sale, and one buyer who has private valuation $x$ for item 1, $y$ for the item 2, and $x+y$ for both items. The seller has zero valuation for any subset of items.

As usual, $x$ and $y$ are unknown to the seller and are treated as independent random variables according to density functions $f_1$ on $[x_A,x_B]\subset \mathbb{R}$ and $f_2$ on $[y_A,y_C]\subset \mathbb{R}$ respectively. The valuation (aka. type) space of the buyer is then $V=[x_A,x_B]\times [y_A,y_C]$. To visualize, we sometimes refer to $V$ as rectangle $ABDC$, where $A$ represents the lowest possible type $(x_A, y_A)$ and $D$ represents the highest possible type $(x_B,y_C)$. Let $f(x,y)=f_1(x)f_2(y)$ be the joint density on $V$. We assume the $f_1$ and $f_2$ are positive, bounded, and differentiable densities.

The seller sells the items through a mechanism that consists of an allocation rule $q$ and a payment rule $t$. In our two-item setting, an allocation rule is conveniently represented by $q=(q_1,q_2)$, where $q_i$ is the probability that buyer gets item $i\in \{1,2\}$. Given valuation $(x,y)$, buyer's utility is \vspace{-0.1in}
$$u(x,y)=xq_1(x,y)+yq_2(x,y)-t(x,y)\vspace{-0.1in}    $$

In other words, buyer has a quasi-linear, additive utility function. It is sometimes convenient to view a mechanism as a (possibly infinite) set of {\em menu items} $\{(q_1(x,y), q_2(x,y), t(x,y))|(x,y)\in [x_A,x_B]\times[y_A,y_C]\}$. Given a mechanism, the expected revenue of the seller is $R=\mathbb{E}_{(x,y)}[t(x,y)]$.


A mechanism must be {\em Individually Rational (IR)}:
\vspace{-0.1in}
$$\forall (x, y), u(x,y)\geq 0.$$
In other words, a buyer cannot get negative utility by participation.

By revelation principle, it is without loss of generality to focus on the set of mechanisms that are {\em Incentive Compatible (IC)}: $$\forall (x, y),(x',y'), u(x,y)\geq xq_1(x',y')+yq_2(x',y')-t(x',y').$$ This means, it is the buyer's (weak) dominant strategy to report truthfully. Equivalently, an IC mechanism presents a set of menu items and let the buyer do the selection (aka. the taxation principle). As a result, $$u(x,y)=sup_{(x',y')}\{xq_1(x',y')+yq_2(x',y')-t(x',y')\},$$ which is the supremum of a set of linear functions of $(x,y)$. Thus, $u$ must be convex.


Fixing $y$, by $IC$, we have
\vspace{-0.1in}
\begin{eqnarray}
&&u(x',y)-u(x,y)-q_1(x,y)(x'-x)\nonumber\\
&=&x'q_1(x',y)+yq_2(x',y)-t(x',y)-xq_1(x,y)-yq_2(x,y)+t(x,y)-x'q_1(x,y)+xq_1(x,y)\nonumber\\
&=&x'q_1(x',y)+yq_2(x',y)-t(x',y)-(x'q_1(x,y)+yq_2(x,y)-t(x,y))\geq 0\nonumber\vspace{-0.1in}
\end{eqnarray}
Substitute $x'$ twice by $x^-=x-\epsilon$ and $x^+=x+\epsilon$ respectively, for any arbitrarily small positive $\epsilon$, we have
$$pu_x(x^-,y)\leq q_1(x,y)\leq u_x(x^+,y),$$ where $u_x$ denotes the partial derivative of $u$ on the $x$ dimension. The inequality above implies $u$ is differentiable almost everywhere on $x$ and $u_x=q_1(x,y)$. Similarly, u is differentiable almost everywhere on $y$ and $u_y=q_2(x,y)$.
As a result $u_x$ and $u_y$ must be within interval $[0,1]$. This means, the seller can never allocate more than one pieces of either item.
Now, payment function $t$ can be represented by utility function $u$, $t(x,y)=xu_x(x,y)+yu_y(x,y)-u(x,y)$.

The seller's problem is to design a non-negative, convex utility function, whose partial derivatives on both $x$ and $y$ are within $[0,1]$, that maximizes expected revenue $R$ (cf.~\cite[Lemma 5]{Hart2012a}).
\vspace{-0.1in}
\section{Representing revenue as a function of utility}

Let $\Omega$ denote any area in $V$ and $R_\Omega$ be the revenue obtained within $\Omega$. Let $z=(x,y)^T$ and $\mathbf{T}(z)=zu(z)f(z)$.

By Green's Theorem, we have
$\int_\Omega \nabla\cdot\mathbf{T}dz=\oint_{\partial \Omega}\mathbf{T}\cdot \mathbf{\hat n}ds.$
\begin{eqnarray}
\nabla\cdot\mathbf{T} &=& 2u(z)f(z)+(\nabla u(z))^Tzf(z)+u(z)z^T \nabla f(z)\nonumber\\
&=&[(\nabla u(z))^Tz-u(z)]f(z)+[3f(z)+z^T \nabla f(z)]u(z)\nonumber\\
&=&t(z)f(z)+\triangle(z)u(z)\nonumber
\end{eqnarray}
where $\triangle(x,y)=3f_1(x)f_2(y)+xf'_1(x)f_2(y)+yf'_2(y)f_1(x)$.


Seller's revenue formula within $\Omega$ is as follows:
\begin{eqnarray}
R_{\Omega}&=&\int_\Omega t(z)f(z)dz=\int_\Omega (\nabla\cdot\mathbf{T}-\triangle(z)u(z))dz\nonumber\\
&=&\oint_{\partial \Omega}\mathbf{T}\cdot \mathbf{\hat n}ds - \int_\Omega \triangle(z)u(z)dz\nonumber
\end{eqnarray}

Set $\Omega$ to be the rectangle $ABDC$, the seller's total revenue $R_{ABDC}$ is
\begin{eqnarray}
& &\int_{y_A}^{y_C}x_Bu(x_B,y)f_1(x_B)f_2(y)dy+\int_{x_A}^{x_B}y_Cu(x,y_C)f_1(x)f_2(y_C)dx\nonumber\\
& &-\int_{y_A}^{y_C}x_Au(x_A,y)f_1(x_A)f_2(y)dy-\int_{x_A}^{x_B}y_Au(x,y_A)f_1(x)f_2(y_A)dx\nonumber\\
& &-\int_{x_A}^{x_B}\int_{y_A}^{y_C}u(x,y)\triangle(x,y)dydx \label{eq1}
\end{eqnarray}

Formula~(\ref{eq1}) consists of 5 terms. The first term represents the part of seller's revenue that depends on utilities on edge $BD$ only. Moreover, this part is increasing as utilities on edge $BD$ increase. Similarly, the second term represents the part of seller's revenue that depends positively on utilities on edge $CD$. The third and fourth terms represent respectively the parts of seller's revenue that depend negatively on utilities on edges $AC$ and $AB$. The fifth term represents the part of revenue that depends on the utilities on the inner points of the rectangle. Under different conditions, $\triangle(x,y)$ can be either positive or negative, which suggests this part can either increase or decrease as the utilities on inner points increases. We now define these conditions.
\vspace{-0.1in}
\begin{definition}
For any density $h(x)$, let $PR(h(x))=\frac{xh'(x)}{h(x)}$ be the {\em power rate} of $h$.
\end{definition}

Consider the following two conditions regarding power rate.
\vspace{-0.1in}
\begin{displaymath}
\mbox{\textbf{Condition 1:}~~}  PR(f_1(x))+PR(f_2(y)) \leq -3, \forall (x,y)\in V.
\end{displaymath}
\begin{displaymath}
\mbox{\textbf{Condition 2:}~~}  PR(f_1(x))+PR(f_2(y)) \geq -3, \forall (x,y)\in V.
\end{displaymath}

Clearly, under Condition 1, we have $\triangle(x,y)\leq 0$. This means seller's revenue depends positively on utilities of the inner points. Similarly, under Condition 2, seller's revenue depends negatively on utilities of the inner points.

Based on the two conditions above, we obtain two parts of results: under Condition 1, the optimal mechanisms have simple menus in the sense that their menus are monotone -- allocation probabilities and payment are increasing in the same order; under Condition 2, the optimal mechanisms also have simple menus, but in a different sense, that their menus only contain a few items.

\cite{Daskalakis13} consider the same problem but restrict to the case where
$$y_Af_2(y_A)=0, x_Af_1(x_A)=0;$$
$$\lim_{x\rightarrow x_B}x^2f_1(x)=0, \lim_{y\rightarrow y_C}y^2f_2(y)=0.$$

These assumptions ignore the effect of utilities on the edges of the rectangle. While we do not have any of these constraints. As a result, their techniques (such as reduction to optimal transport) do not apply to our more general case. In fact, one of our main techniques is to show how the optimal revenue changes as a function of the utilities on the edges.
\section{Part I: menu monotonicity and revenue monotonicity}
\label{sec:part1}
In this section, we consider the case where power rates of both density functions satisfy Condition 1. When this condition is not met,~\citet{Hart2012b} give several interesting counter-examples of revenue monotonicity: the optimal revenue for stochastically inferior valuation distributions may be greater than that of stochastically superior distributions. When this condition is met, for identical item distributions,~\citet{Hart2012a} prove that, bundling sale is the optimal mechanism. In this section, we show that, under Condition 1, the optimal menu can be sorted so that both allocations as well as payment monotonically increase. We coin this result {\em menu monotonicity theorem}. The theorem has two immediate consequences. First, it yields a version of revenue monotonicity theorem, thus complements the Hart-Reny result above. Second, it subsumes the above Hart-Nisan result as a corollary.

Our analysis starts from a simple observation: any optimal mechanism must extract all of the buyer's valuation when he is in the lowest type.

\begin{lemma}
In the optimal mechanism, $u(x_A,y_A)=0$.
\label{lemma:revenuemonotone1}
\end{lemma}

\begin{proof}
Suppose otherwise that $u(x_A,y_A)>0$, one can revise every menu item from $(q_1(x,y),q_2(x,y),t(x,y))$ to $(q_1(x,y),q_2(x,y),t(x,y)+u(x_A,y_A))$ and obtain a mechanism with strictly higher revenue, contradiction.
\end{proof}

\begin{theorem}
\label{theorem:monotone}
\textbf{Menu Monotonicity}\\
Under Condition 1, each menu item of the optimal mechanism can be labeled by a real number $s$: $(q_1(s),q_2(s),\\t(s))$, such that $q_1(s),q_2(s)$ is weakly increasing and $t(s)$ are strictly increasing in $s$.
\label{theorem:menumonotone}
\end{theorem}

Roughly speaking, Theorem~\ref{theorem:monotone} suggests that, among the menu items of the optimal mechanism, higher $t$ corresponds to higher $q_1$ and $q_2$. Note that allocation and payment monotonicity are well understood in single-item optimal auction (i.e., Myerson auction) but in general fail to hold in two item settings~\citep{Hart2012b}.

In the following, we give a semi-constructive proof. By Formula~(\ref{eq1}), under Condition 1, we know that seller's revenue is increasing as the utilities of the buyer increases on $V$, only except on edges $AB$ and $AC$. Our idea is then, to fix the optimal utility function on $AB$ and $AC$ and construct the (largest possible) remainder of the utility function subject to convexity.

\begin{proof}
By Lemma~\ref{lemma:revenuemonotone1}, $u(x_A,y_A)=0$.
In the following, we start from the {\em optimal} utility function $u$. Let $u(AB)$ denote the part of $u$ on edge $AB$. Similar for $u(AC)$.

Let $u^{(x_0,y_0),0}(x,y)=xq_1(x_0,y_0)+yq_2(x_0,y_0)-t(x_0,y_0)$ for any $(x_0,y_0)$. In other words, $u^{(x_0,y_0),0}(x,y)$ is the buyer's
 utility at type $(x,y)$ but chooses menu item $(q_1(x_0,y_0)$, $q_2(x_0,y_0)$, $t(x_0,y_0))$. By IC, $u^{(x_0,y_0),0}(x,y)\leq u(x,y)$. To
 visualize, think of $u^{(x_0,y_0),0}(x,y)$ as a plane that is always weakly below $u(x,y)$ but touches $u$ at the point of $(x_0, y_0)$.

Apply the following two-step operation:
\begin{figure}[h]
  \centering
  \includegraphics[width=7cm]{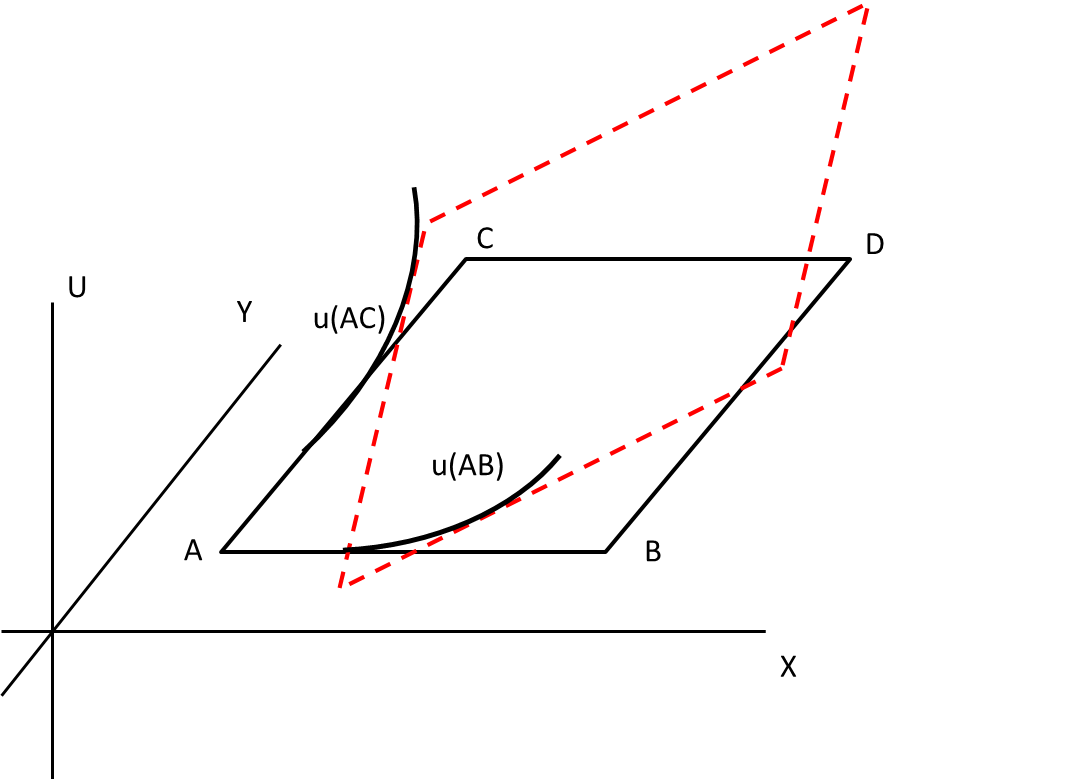}\\
  \caption{A menu item in the optimal mechanism}
  \label{fig:example0}
\end{figure}
\begin{itemize}
\item Step 1. Translation. Rise up the plane of $u^{(x_0,y_0),0}$ uniformly for every $(x,y)$ until it touches one of $u(AB)$ and $u(AC)$, say $u(AB)$.
Denote the plane after Step 1 as $u^{(x_0,y_0),1}$ and notice the part of $u^{(x_0,y_0),1}$ on $AB$, which is a straight line.
\item Step 2. Rotation. Fix $u^{(x_0,y_0),1}(AB)$, rotate the plane $u^{(x_0,y_0),1}$ until it touches $u(AC)$.(We must also consider the case where the
slope of $u^{(x_0,y_0),1}(AC)$ reaches 1 but still does not touches $u(AC)$. If this occurs, simply return the current plane.) Denote the plane after
Step 2 as $u^{(x_0,y_0),2}$. In Fig.~\ref{fig:example0}, the red dashed plane represents plane $u^{(x_0,y_0),2}$, it touches both $u(AB)$ and $u(AC)$.
\end{itemize}

We have $u^{(x_0,y_0),2}(x_0,y_0)\geq u^{(x_0,y_0),1}(x_0,y_0)\geq u^{(x_0,y_0),0}(x_0,y_0)=u(x_0,y_0)$.
Repeat the same procedure for all $(x_0,y_0)$ and define another function $u^*=sup_{(x_0,y_0)} \{u^{(x_0,y_0),2}\}$. In other words, $u^*$ is the supremum
of all the planes after the two-step operation.

We claim that, when the $u$ is optimal, $u^*$ must be optimal as well. In fact, by the construction of $u^*$, we know that $u^*(x,y)\geq u^{(x,y),2}(x,y)
\geq u(x,y)$ for any $(x,y)$. Also, by the construction of $u^{(x,y),2}$, for any $(x_0,y_0)$ on $AB$ or $AC$, $u^{(x_0,y_0),2}(x_0,y_0)=u(x_0,y_0)$
and for any $(x,y)$, $u^{(x,y),2}(x_0,y_0)\leq u(x_0,y_0)$ since any plane $u^{(x,y),2}$ resulted from the two-step operation must weakly remain
below $u(AB)$ and $u(AC)$. Thus, by the definition of $u^*$, $u^*$ and $u$ are identical on $AB$ and $AC$. So $u^*$ will lead to a weakly greater
revenue than u according to Formula(1). Thus, $u^*$ must be optimal as well. This is sufficient to establish that $u^*$ must consist of a set of planes.

Pick one arbitrary plane $\hat{u}$ among $u^*$, it intersects with the u-axis at some point $k=(0,0,k)$, we assume at the same time it intersects with
vertical line $(x=x_A,y=y_A)$ at some point $(x=x_A,y=y_A,u=u_k)$($u_k$ is zero or negative).

We now show that $\hat{u}$ is the unique plane among $u^*$ that goes through point $(x_A,y_A,u_k)$. Consider $\hat{u}\cap (y=y_A)$(intersection line
of two plane $\hat{u}$ and plane $y=y_A$), this line is unique. Moreover, it either has gradient 1 or touches $u(AB)$. So the gradient, say $q_1$, of
this line is unique. Recall that $q_1$ is also the gradient of plane $\hat{u}$ in the $x$-dimension. So the gradient of plane $\hat{u}$ in the
$x$-direction is unique.
For same reason, $q_2$, the gradient of plane $\hat{u}$ in $y$ direction is also unique. Because there is only one plane that goes through point
$(x_A,y_A,u_k)$ can satisfy the two gradient conditions, plane $\hat{u}$ can be uniquely determined by $u_k$.

Finally, lower $u_k$ results in weakly larger $q_1$ and $q_2$ (by convexity of $u(AC)$ and $u(AB)$), strictly lower $k=z_k-q_1x_A-q_2y_A$, and
hence strictly larger payment $t=|k|$. In other words, larger $t$ indeed corresponds to steeper $\hat{u}(AB)$ and $\hat{u}(AC)$, hence larger
$q_1$ and $q_2$. This completes the proof of menu monotonicity.
\end{proof}
\vspace{-0.1in}
Theorem~\ref{theorem:monotone} implies the aforementioned Hart-Nisan result as a corollary.
\begin{corollary}
\vspace{-0.05in}
~\cite[Theorem 28]{Hart2012a} For two i.i.d. items, $PR(f_1)=PR(f_2)\leq -\frac32$, bundling sale is optimal.
\end{corollary}
\vspace{-0.1in}
\begin{proof}
It is without loss to restrict attention to symmetric mechanisms~\citep{Maskin1984}. Thus, $u(AB)$ is identical to $u(AC)$. $u^1(AB)$ and $u^2(AC)$ have the same slope (in fact, plane $u'$ touches both $u(AB)$ and $u(AC)$ simultaneously). So, $q_1(v)=q_2(v)$ $\forall v\in V$. In other words, the two items are always sold with the same probability. The seller's revenue of this optimal mechanism is equivalent to a mechanism that sells two items as a bundle with the same probability. So bundling is optimal as well.
\end{proof}
\vspace{-0.1in}
As another application of Theorem~\ref{theorem:monotone}, we obtain a \emph{revenue monotonicity} theorem in this setting.
\begin{theorem}
\textbf{(Revenue Monotonicity)}\\
Under Condition 1, optimal revenue is monotone: let $F_i,G_i$ be the cumulative distribution function of density functions $f_i,g_i,i=1,2$, respectively. If $G_1$ and $G_2$ first-order stochastically dominate $F_1$ and $F_2$\footnote{$G_i$ first-order stochastically dominates $F_i$ if $G_i(x)\leq F_i(x)$ for all $x$ and $G_i(x)> F_i(x)$ for some $x$.} respectively, optimal revenue obtained for $(G_1, G_2)$ is no less than that of $(F_1, F_2)$.
\end{theorem}
\begin{proof}
Consider any two points $(x_2,y_2)$ and $(x_3,y_2)$, where $x_3>x_2$. If $q_1(x_2,y_2)<q_1(x_3,y_2)$, by Theorem \ref{theorem:menumonotone}, we must have $t(x_2,y_2)<t(x_3,y_2)$.
If $q_1(x_2,y_2)=q_1(x_3,y_2)$, then $q_1(x,y_2)=q_1(x_2,y_2), \forall x\in[x_2,x_3]$.  $u(x_3,y_2)=u(x_2,y_2)+q_1(x_2,y_2)(x_3-x_2)$, which can be achieved by choosing the same menu as $(x_2,y_2)$ chooses.
While buyer at type $(x_3,y_2)$ has several menu items that all achieve the highest utility, we can assume the buyer chooses the menu with the highest payment
(\citep{Hart2012b}).

Thus there is an optimal choice guarantees $t(x_2,y_2)\leq t(x_3,y_2)$.

To sum up, $t(x_2,y_2)\leq t(x_3,y_2)$ when $x_2\leq x_3$. For same reason, $t(x_3,y_2)\leq t(x_3,y_3)$ when $y_2\leq y_3$. Hence $t(x,y)$ is a weakly monotone function in both directions.
Suppose $G_1$ and $G_2$ first-order stochastically dominates $F_1$ and $F_2$ respectively. Let $R(F_1\times F_2)$ denote the optimal revenue when item $1$ and $2$ distributes independently according to $F_1$ and $F_2$. When distribution $G_1\times G_2$ chooses the same mechanism as $F_1\times F_2$ does, let the revenue be $R^*(G_1\times G_2)$. We have $R^*(G_1\times G_2)\geq R(F_1\times F_2)$, since $t$ is weakly monotone. By transitivity, $R(G_1\times G_2)\geq R^*(G_1\times G_2)\geq R(F_1\times F_2)$.
\end{proof}
\section{Part II: Optimal mechanism with small menus}
\vspace{-0.1in}
\label{sec:part2}
In this section, we investigate optimal mechanisms under Condition 2. We obtain several results saying that the optimal mechanism contains only few menu items. All these results are built upon Pavlov's characterization~\citep{Pavlov2011a} and an important lemma introduced in the next subsection.
\vspace{-0.12in}
\subsection{Pavlov's characterization and graph representation lemma}
\label{subsec:part21}


If both $f_1$ and $f_2$ satisfy Condition 2,~\citet[Proposition 2]{Pavlov2011a} states that, in the optimal mechanism, the seller either keeps both items (i.e., $q=(0,0)$), or sells one of the items at probability 1 (i.e., $q_1=1$ or $q_2=1$).

For graphic representation, let the buyer's valuation be within rectangle $ABDC$, we have the following lemma.
\vspace{-0.05in}
\begin{lemma} \vspace{-0.1in}  \textbf{Graph Representation Lemma}\\
\label{lemma:graph}
\begin{figure}
\setlength{\abovecaptionskip}{-0.6cm}
\setlength{\belowcaptionskip}{-0.5cm}
  \centering
  \begin{tabular}{cc}
  \begin{minipage}[t]{3.2in}
  \includegraphics[width=7cm]{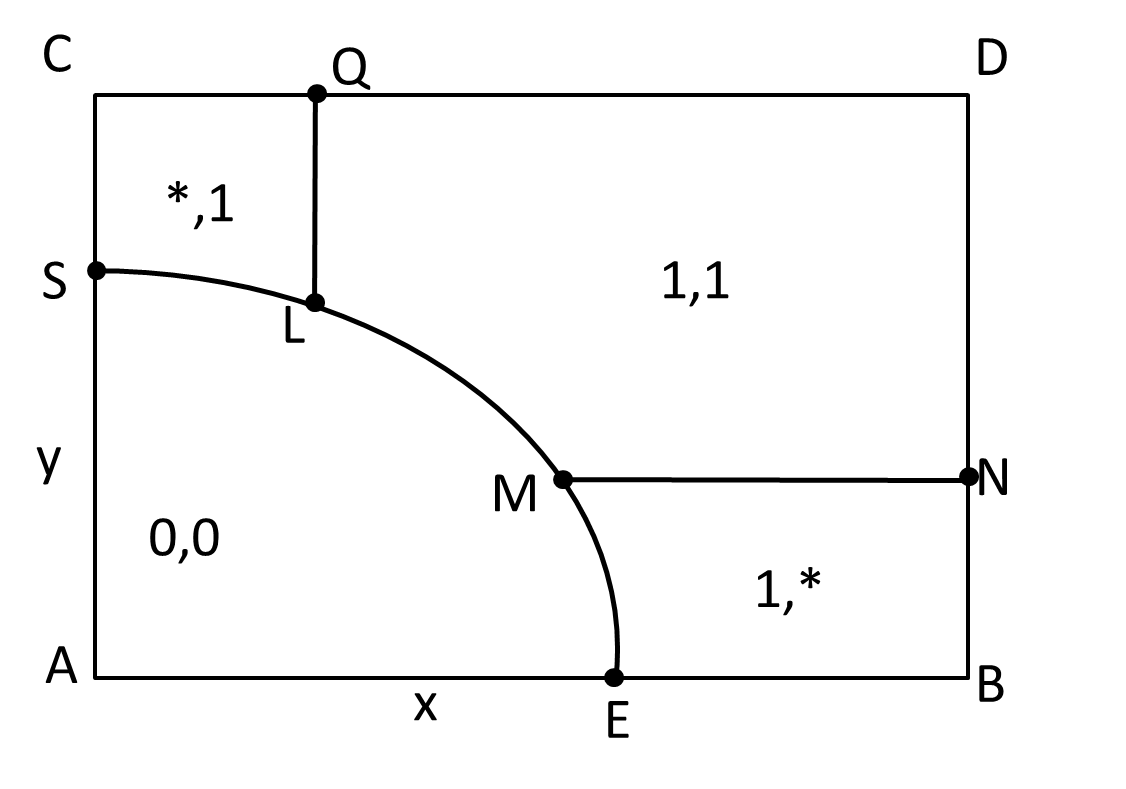}\\
  \caption{The optimal allocation that there is a point on curve SME chooses allocation menu (1,1).}
  \label{fig:example1}
  \end{minipage}
  \begin{minipage}[t]{3.2in}
  \includegraphics[width=7cm]{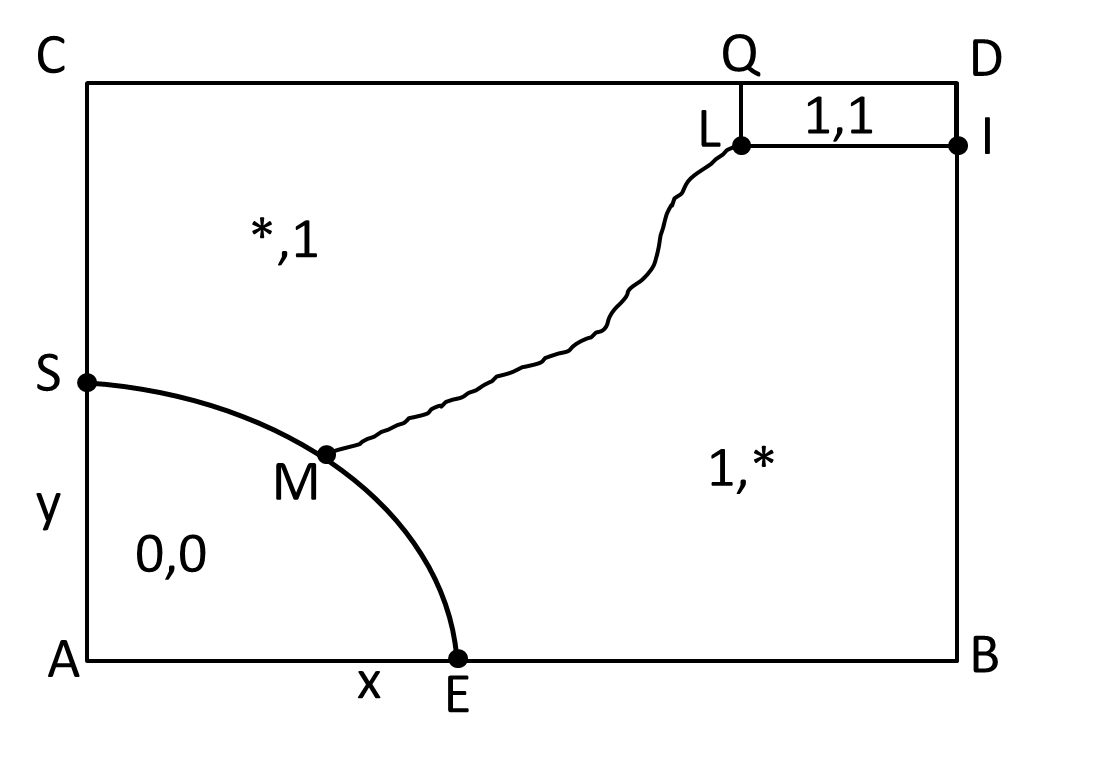}\\
  \caption{The optimal allocation that there is no point on curve SME chooses allocation menu (1,1).}
  \label{fig:example2add}
  \end{minipage}
  \end{tabular}
\end{figure}
\vspace{-0.05in}
Under Condition 2, the optimal mechanism can be represented by one of the rectangles shown in Fig.~\ref{fig:example1} or Fig.~\ref{fig:example2add}. More precisely, the optimal mechanism divides the valuation space into four regions, where
\begin{enumerate}
\item in the bottom left region (region $ASME$ in both figures), it assigns $q=(0,0)$ and $u(x,y)=0$ to any point $(x,y)$ in the region. Furthermore, region $ASME$ is convex.
    \vspace{-0.1in}
\item in the top right region, it assigns $q=(1,1)$ to any point in the region.
\vspace{-0.1in}
\item in the top left region, it assigns $q=(*,1)$ to any point in the region, where $*$ is a variable. Thus this region represents a set of menu items, each of which is a vertical slice.
    \vspace{-0.1in}
\item Symmetrically, in the bottom right region, it assigns $q=(1,*)$ to any point in the region. This region represents a set menu items, each of which is a horizontal slice.
    \vspace{-0.1in}
\item The boundary between the top left and right regions is vertical ($QL$ in both figures); the boundary of the top right and bottom right regions is horizontal ($MN$ in Fig.~\ref{fig:example1} or $LI$ in Fig.~\ref{fig:example2add}). The boundary between $(1,*)$ region and $(*,1)$ region is in the upper right direction.
\end{enumerate}
\label{lemma:pre}
\end{lemma}
\vspace{-0.3in}
\subsection{Optimal mechanisms for constant power rate}
\label{subsec:part22}

To describe our first theorem under Condition 2, we need the following condition on density functions.
\vspace{-0.1in}
\begin{displaymath}\vspace{-0.1in}
\mbox{\textbf{Condition 3:}~~} PR(f_i(x)), ~~i=1,2,~\textrm{ is constant.}
\end{displaymath}
\begin{theorem}
\label{thm:4items}
Under Conditions 2 and 3, optimal mechanism has at most 4 menu items.
\end{theorem}

The result is tight because one can find instances where optimal mechanism contains exactly 4 menu items~\cite[Example 3]{Pavlov2011a}.

We prove the theorem for the case shown in Fig.~\ref{fig:example2}. The other case related to Fig.~\ref{fig:example1} follows from an almost identical proof. First, draw a horizontal line through $M$ and it intersects $BD$ at $N$.
Then draw a vertical line through $M$ crossing $CD$ at $G$. We have the following two lemmas.
\vspace{-0.1in}
\begin{lemma}
Optimal utility function on $BN$ is piecewise linear with at most two pieces.
\label{lemma:fourcanon}
\vspace{-0.05in}
\end{lemma}
\vspace{-0.05in}
\begin{figure}
\setlength{\abovecaptionskip}{-0.4cm}
\setlength{\belowcaptionskip}{-0.5cm}
  \centering
  \includegraphics[width=6cm]{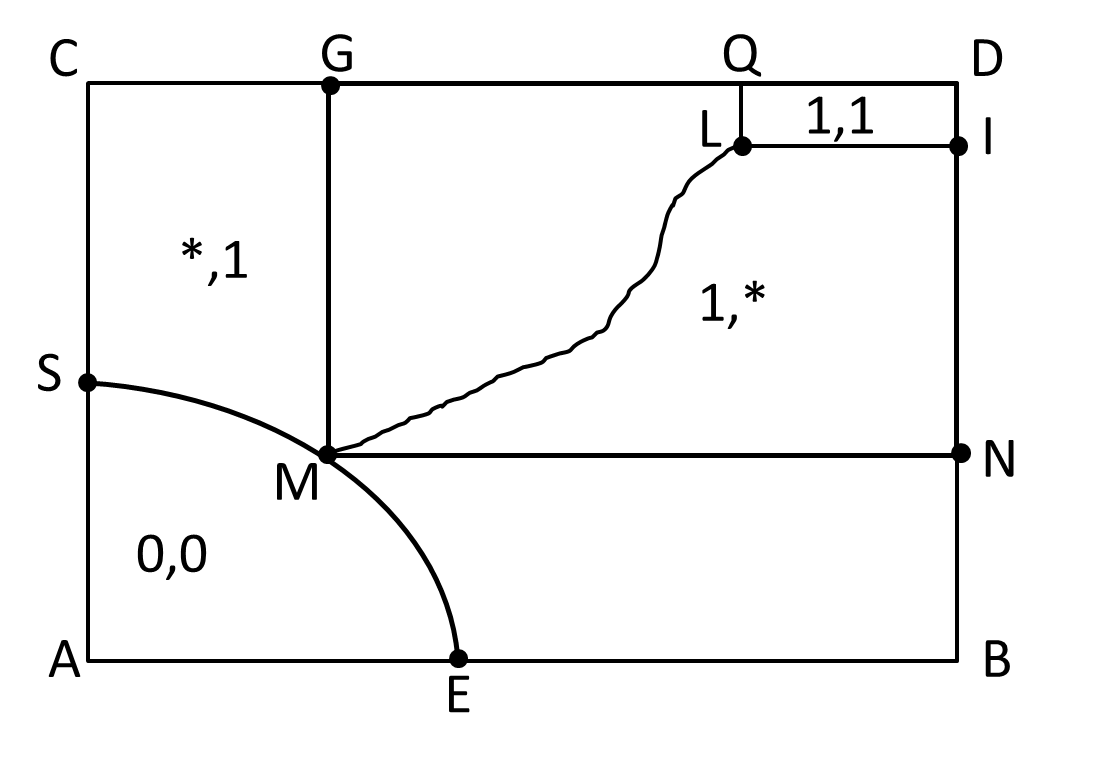}\\
  \caption{The optimal allocation that there is no point on curve SME chooses allocation menu (1,1).}
  \label{fig:example2}
\end{figure}
\begin{lemma}
The optimal utility function on $ND$ is piecewise linear with at most 2 pieces.
\label{lemma:four2}
\end{lemma}

With these two lemmas, we are able to prove Theorem~\ref{thm:4items} (in the appendix).
 \citet[Theorem 1 and Lemma 14]{Hart2012a} state that bundling 4-approximates optimal revenue for general two-item setting. As an application of Theorem~\ref{thm:4items}, we obtain a better lower bound for bundling sale.
\begin{corollary}
Under Conditions 2 and 3, bundling 3-approximates optimal revenue.
\end{corollary}
\begin{proof} 
Revenue of an optimal mechanism with 3 non-zero menu items is less than or equal to the sum of revenues of 3 mechanisms, each of which has only 1 non-zero menu items. Since bundling is optimal among all mechanisms that contains only 1 non-zero menu item, thus no worse than any of these three mechanisms. Consequently bundling gives a 3-approximation of the optimal revenue.\vspace{-0.1in}
\end{proof}

Following a similar proof of Theorem~\ref{thm:4items}, we obtain another condition under which 3 menu items are enough. Note that, this condition does not impose constant power rate, thus is not a special case of Condition 3.
\begin{displaymath}
\mbox{\textbf{Condition 4:}~~} -2\leq PR(f_1(x))\leq y_Af_2(y_A)-2, \forall x\textrm{ and } -2\leq PR(f_2(y))\leq x_Af_1(x_A)-2,\forall y.
\end{displaymath}
\begin{theorem}
Under Conditions 2 and 4, there is an optimal mechanism such that it contains at most 3 menu items, thus bundling gives a 2-approximation of the optimal revenue.
\label{theorem:6}
\end{theorem}

\subsection{Optimal mechanisms for monotone power rate}
\label{subsec:part24}

The requirement of power rate to be constant might be restrictive. If one relaxes this requirement to be \emph{monotone power rate}, one only needs to add two more menu items.
\vspace{-0.1in}
\begin{displaymath}
\mbox{\textbf{Condition 5:}~~} PR(f_i(x)), i=1,2, \textrm{is weakly monotone.}
\end{displaymath}
\begin{theorem}
\vspace{-0.1in}
Under Conditions 2 and 5, if $f_1=f_2$, optimal mechanism consists of at most 6 menu items.
\label{theorem:2}
\end{theorem}
\vspace{-0.1in}
The general form of optimal mechanism is shown in Fig~\ref{fig:example4}.
It is without loss to restrict attention on symmetric mechanisms~\cite[section 1]{Maskin1984}. Let $AD$ intersects $SE$ at point $M$. In region $ASME$, seller keeps both items.
Item 2 is sold deterministically in $CSMD$ and item 1 is sold determinately in $MEBD$.
Let the allocation rule on point $(x,y)$ in $CSMD$ be $(q_1(x),1)$. Similar to the proof of Theorem~\ref{thm:4items}, we start with the following lemma.
\vspace{-0.1in}
\begin{lemma}
Optimal utility function on $ND$ is piecewise linear with at most 2 pieces.
\label{lemma:monotonePR2}
\end{lemma}

\begin{figure}
\setlength{\abovecaptionskip}{-0.2cm}
\setlength{\belowcaptionskip}{-0.2cm}

  \begin{tabular}{cc}
  \begin{minipage}[t]{3.2in}
  \centering
  \includegraphics[width=5cm]{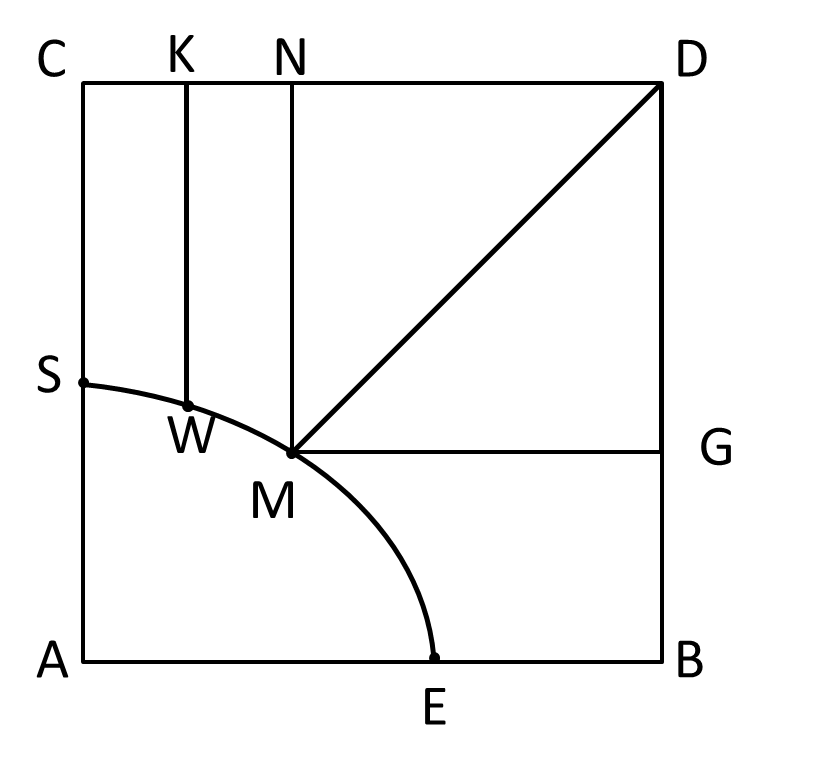}\\
  \caption{Optimal allocation under symmetric value distribution that satisfies Condition 5.}
  \label{fig:example4}
  \end{minipage}
  \begin{minipage}[t]{3.2in}
  \centering
  \includegraphics[width=6cm]{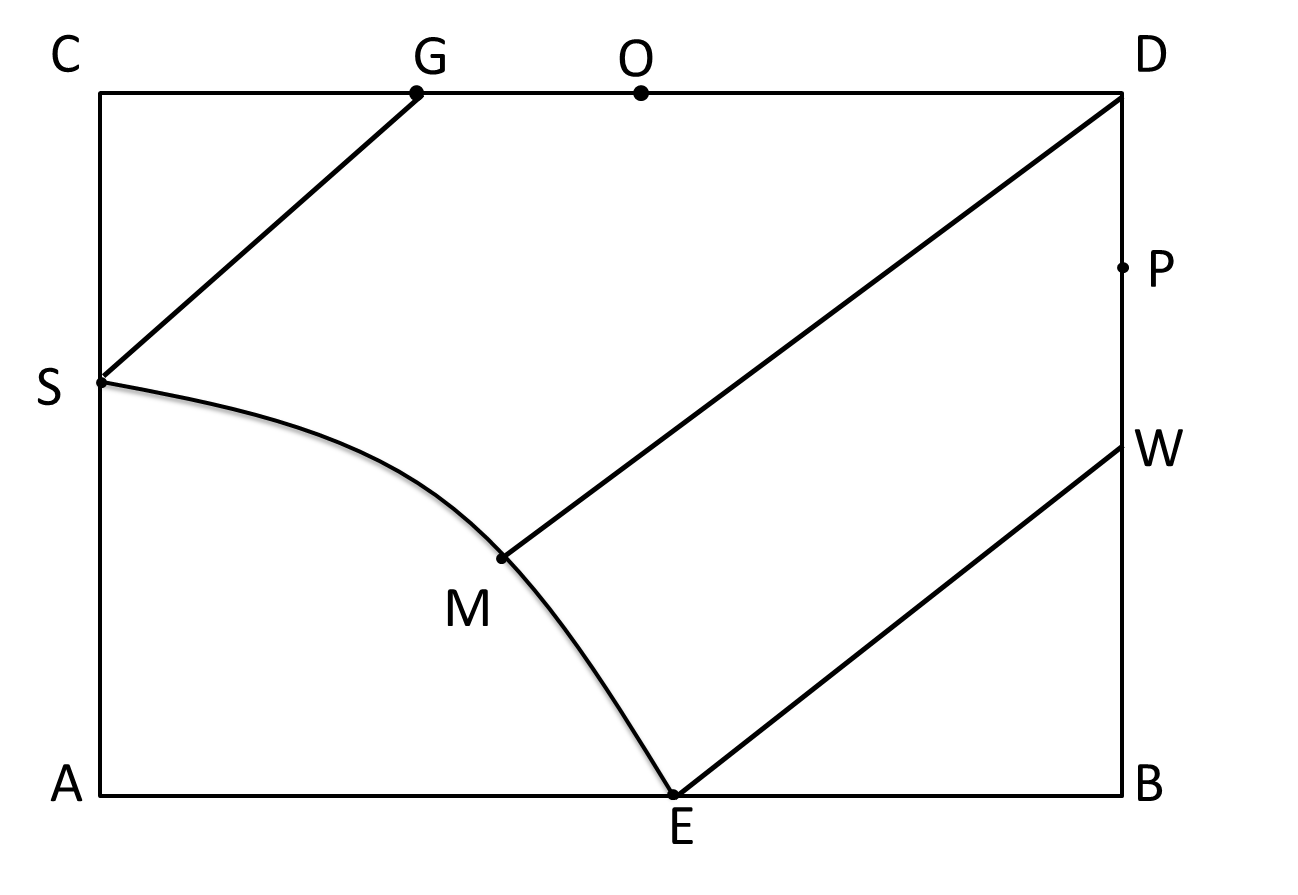}\\
  \caption{Optimal unit demand allocation.}
  \label{fig:example6}
  \end{minipage}
  \end{tabular}
\end{figure}
\vspace{-0.05in}
Similarly, we show that $u(CN)$ is piecewise linear as well.
\vspace{-0.05in}
\begin{lemma}
Optimal utility function on $CN$ is piecewise linear with at most 2 pieces.
\label{lemma:monotonePR3}
\end{lemma}
\vspace{-0.05in}
With the two lemmas above, we are able to prove Theorem~\ref{theorem:2}.

We say $h(x)$ are nonnegative-coefficient polynomial if $h(x)=a_nx^n+a_{n-1}x^{n-1}+...+a_1x+a_0, a_i\geq 0, i=0,...,n$.
Let $h_1$ and $h_2$ are nonnegative-coefficient polynomials, it is easy to show that $h_1(x)e^{h_2(x)}$ satisfies Condition 5. In particular, this expression includes nonnegative-coefficient polynomial density and exponential density as special cases.
\vspace{-0.1in}
\subsection{Optimal mechanism for uniform distributions under unit-demand constraint}
\label{subsec:part25}

A buyer has {\em unit-demand} if $q_1(x,y)+q_2(x,y)\leq 1$. Under unit-demand model, ~\citet[Proposition 2]{Pavlov2011b} states that, if distribution functions satisfy Condition 2, it is without loss to restrict attention on mechanisms such that
$$q_1(x,y)+q_2(x,y)\in\{0,1\}\qquad \forall (x,y)$$
Pavlov solves the optimal mechanism for two items with identical uniform distributions. The resulting mechanism contains 5 menu items for uniform distribution on $[c,c+1]*[c,c+1], c \in(1,\bar{c})$ (where $\bar{c} \approx 1.372$).
We show that in nonidentical settings, the optimal mechanism also contains at most 5 menu items. It follows trivially that our result is tight.
\vspace{-0.05in}
\begin{theorem}
In unit-demand model, if both $f_1$ and $f_2$ are uniform distributions, the optimal mechanism consists of at most 5 menu items.
\label{thm:uni}
\end{theorem}
\vspace{-0.05in}
Let $ASE$ denote the zero utility region and $CSEBD$ the non-zero utility region. For the same reason in Lemma \ref{lemma:pre}, $ASE$ is convex. For points in $ASE$, allocation $(0,0)$ is the best. For $(x,y)\in CSEBD$, $(q_1(x,y),q_2(x,y))\neq (0,0)$, so $q_1(x,y)+q_2(x,y)=1$.
The mechanism is shown in Fig.~\ref{fig:example6}.
Draw a $45$ degree line across $E$, intersecting $BD$ or $CD$ at $W$.
Draw a $45$ degree line across $S$, intersecting $BD$ or $CD$ at $G$.
We consider here the case that $W$ is on $BD$ and $G$ is on $CD$. Other cases follow from similar arguments.

The theorem can be similarly proved via the following two lemmas.
\vspace{-0.1in}
\begin{lemma}
Optimal utility function on $BW$ is piecewise linear with at most 2 pieces.
\vspace{-0.05in}
\label{lemma:uni1}
\end{lemma}
\begin{lemma}
Optimal utility function on $WD$ is piecewise linear with at most 2 pieces.
\label{lemma:uni2}
\end{lemma}



\bibliographystyle{plainnat}
\bibliography{team}



\begin{appendix}

\section{Appendix: Proofs of Section~\ref{sec:part2}}

\begin{proof}\textbf{of Lemma~\ref{lemma:graph}}We first determine the relative positions of the four possible regions.

If the seller keeps both items, the buyer's utility is zero. Since $u(x,y)$ is an increasing function, it assigns $q=(0,0)$ in the bottom left region, i.e. $ASME$. Since $u(x,y)$ is convex, the convex combination of any two zero-utility points must also be zero. Therefore, $ASME$ is a convex region.

If for a type $(x_0,y_0)$ with $q_1(x_0,y_0)=q_2(x_0,y_0)=1$, for any point $(x_1,y_1)$, $IC$ requires that
\begin{eqnarray}
x_0+y_0-t(x_0,y_0)\geq x_0q_1(x_1,y_1)+y_0q_2(x_1,y_1)-t(x_1,y_1),\nonumber\\
x_1q_1(x_1,y_1)+y_1q_2(x_1,y_1)-t(x_1,y_1)\geq x_1+y_1-t(x_0,y_0).\nonumber
\end{eqnarray}
Summing the two inequalities, we get $(q_1(x_1,y_1)-1)(x_1-x_0)+(q_2(x_1,y_1)-1)(y_1-y_0)\geq 0$.
If $x_1>x_0, y_1>y_0$, we must have $q_1(x_1,y_1)=q_2(x_1,y_1)=1$.

Let $(x_2,y_2)$ be a point where some positive proportions of the items are sold, then according to Pavlov's characterization~\citep{Pavlov2011a}, one of the items must be sold deterministically. Consider two types $(x_2,y_2)$ and $(x_3,y_3)$ where $q_1(x_2,y_2)=1, q_2(x_2,y_2)<1$ and $q_1(x_3,y_3)<1,q_2(x_3,y_3)=1$.
By $IC$, we must have
\begin{eqnarray}
x_2+y_2q_2(x_2,y_2)-t(x_2,y_2)\geq x_2q_1(x_3,y_3)+y_2-t(x_3,y_3),\nonumber\\
x_3q_1(x_3,y_3)+y_3-t(x_3,y_3)\geq x_3+y_3q_2(x_2,y_2)-t(x_2,y_2).\nonumber
\end{eqnarray}
Summing up the two inequalities, we get $(1-q_1(x_3,y_3))(x_2-x_3)+(1-q_2(x_2,y_2))(y_3-y_2)\geq 0$. So, one of $x_2<x_3$ and $y_2>y_3$ does not hold. This implies the second part of (5).

To sum up, $(1,1)$ must be assigned to the upper right corner, $(1,q_2(x,y))$ is assigned to the bottom right corner, $(q_1(x,y),1)$ is assigned to the upper left corner, and $(0,0)$ is assigned to the bottom left corner, (some regions may be empty).

Let the allocation vector at $(x_4,y_4)$ be $(1,q_2(x_4,y_4))$. For any $x\in[x_4, x_B]$, by IC, we must have $q_1(x,y_4)=1$, so $u(x,y_4)=u(x_4,y_4)+x-x_4$. It is still in the buyer's best interest to choose menu item $(1,q_2(x_4,y_4))$ at $(x,y_4)$. This immediately implies the first part of (5): the boundary between for different $q_2$ in $(1,q_2(x,y))$ is horizontal. In particular, in Fig.~\ref{fig:example1}, $MN$ is horizontal. Similarly, $LQ$ is vertical.

If there is a point on curve $SE$ that chooses menu item $(1,1)$, the mechanism is of the form shown in Fig.~\ref{fig:example1}, otherwise it is of the form shown in Fig.~\ref{fig:example2add}.
\end{proof}

\begin{proof}\textbf{of Lemma~\ref{lemma:fourcanon}}
For any point $K$ on $BN$(see Fig.\ref{fig:newproof}), draw a horizontal line that intersects curve $SE$ at $W$. Because $u_x=1$ for all the types in region $EBNM$ , we can represent the utility of any point $(x,y_K)$ on line $KW$ using $u(K)$. Formally, $u(x,y_K)=u(K)+x-x_K, \forall x\in[x_W,x_K]$.

The revenue obtained in $EBNM$ is
\begin{figure}
  \centering
  \includegraphics[width=7cm]{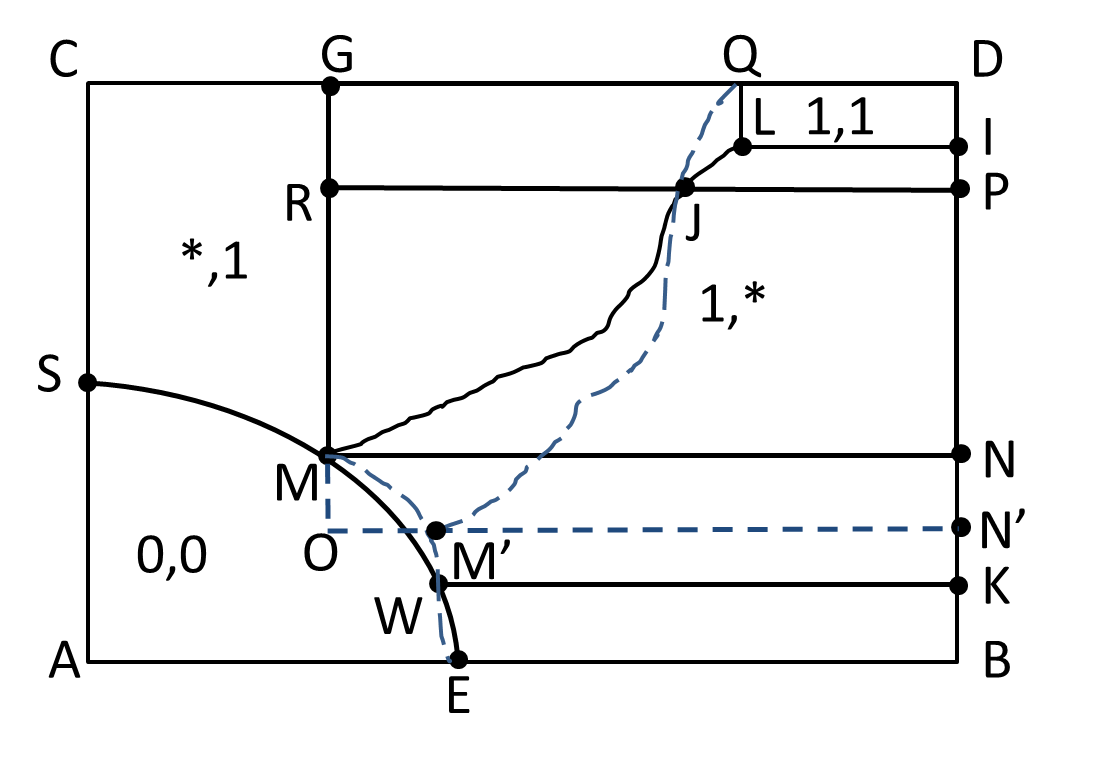}\\
  \caption{The optimal allocation that there is no point on curve SME chooses allocation menu (1,1).}
  \label{fig:newproof}
\end{figure}
\begin{eqnarray}
R_{EBNM}&=&\oint_{EBNM}\mathbf{T}\cdot \mathbf{\hat n}ds - \int_{EBNM} \triangle(z)u(z)dz\nonumber\\
&=&\int_{BN}\mathbf{T}\cdot \mathbf{\hat n}ds -\int_{y_A}^{y_M}\int_{x_B-u(x_B,y)}^{x_B}\triangle(x,y)u(x,y)dxdy+\int_{NMEB}\mathbf{T}\cdot \mathbf{\hat n}ds\nonumber\\
&=&\int_{y_A}^{y_M}[x_Bu(x_B,y)f_1(x_B)f_2(y)-\int_{x_B-u(x_B,y)}^{x_B}\triangle(x_B,y)(u(x_B,y)+x-x_B)dx]dy+c\nonumber\\
&=&\int_{y_A}^{y_M}R(u(x_B,y),y)dy+c\nonumber
\end{eqnarray}
$c=\int_{NMEB}\mathbf{T}\cdot \mathbf{\hat n}ds$ depends on $u(NM),u(ME)$,and $u(EB)$. We fix these utility and study $u(BN)$.
Let $R(u(x_B,y),y)=x_Bu(x_B,y)f_1(x_B)f_2(y)-\int_{x_B-u(x_B,y)}^{x_B}\triangle(x_B,y)(u(x_B,y)+x-x_B)dx$ and
$R_u(u(x_B,y),y)$ denote the partial derivative on $u(x_B,y)$ dimension.
$$R_u(u(x_B,y),y)=f_2(y)[x_Bf_1(x_B)-\int_{x_B-u(x_B,y)}^{x_B}f_1(x)(3+PR(f_1(x))+PR(f_2(y)))dx]$$

Let $v(l)=\int_{x_B-l}^{x_B}f_1(x)(3+PR(f_1(x))+PR(f_2(y)))dx$. Because $3+PR(f_1(x))+PR(f_2(y))\geq 0$ is a constant, $v$ is an increasing function of $l$. When $x_B f_1(x_B)-v(u(x_B,y))>0$, $R_{EBNM}$ is increasing as $u(x_B,y)$ increases. From the Fig.\ref{fig:newproof} we can see, when $(x_B,y)$ moves from $B$ to $N$, $u(x_B,y)$ weakly increases, $v(u(x_B,y))$ weakly increases and $x_Bf(x_B)-v(u(x_B,y))$ weakly decreases.

There are 3 possibilities: 1.$x_Bf(x_B)-v(u(x_B,y))> 0, y\in[y_B,y_N]$, 2.$x_Bf(x_B)-v(u(x_B,y))< 0, y\in[y_B,y_N]$, 3.$\exists y'\in[y_B,y_N], x_Bf(x_B)-v(u(x_B,y))=0$ and $x_Bf(x_B)-v(u(x_B,y))\leq 0, y\in[y_B,y'], x_Bf(x_B)-v(u(x_B,y))\geq 0, y\in[y',y_N]$. 1 can be regarded as a special case of 3 by setting $y'=y_B$, because $u(b,y), y\in[y_B,y_N]$ must be also as large as possible. 2 can be regarded as an special case of 3 by setting $y'=y_N$, because $u(x_B,y), y\in[y_B,y_N]$ must be as small as possible.

So it is without loss to restrict attention on case 3. Since $K$ is randomly chosen, set $y'=y_K$. Revenue is increasing as $u(KN)$ decrease and revenue is increasing as $u(KB)$ increase. By convexity of $u$, with fixed $u(B)$, $u(K)$, and $u(N)$,
optimal $u(BN)$ comprises of two lines: the straight line(with extended line) across points $(y_B,u(B))(y_K,u(K))$, and straight line across point $(y_N,u(N))$ with slope $q_2(N)$.
\end{proof}

\begin{proof}\textbf{of Lemma~\ref{lemma:four2}}
The revenue obtained in region $MNDG$ is
\begin{eqnarray}
R_{MNDG}&=&\oint_{MNDG}\mathbf{T}\cdot \mathbf{\hat n}ds - \int_{MNDG} \triangle(z)u(z)dz\nonumber\\
&=&\int_{DGMN}\mathbf{T}\cdot \mathbf{\hat n}ds + \int_{ND}\mathbf{T}\cdot \mathbf{\hat n}ds-
\int_{y_N}^{y_D}\int_{x_M}^{x_N} \triangle(x,y)u(x,y)dxdy\nonumber\\
&=&c + \int_{y_N}^{y_D}[x_Bu(x_B,y)f_1(x_B)f_2(y) - \int_{x_M}^{x_N} \triangle(x,y)u(x,y)dx]dy\nonumber
\end{eqnarray}
$c=\int_{DGMN}\mathbf{T}\cdot \mathbf{\hat n}ds$ depends on $u(DG),u(GM)$, and $u(MN)$. We fix these utilities and study $u(DN)$. Let $R(u(x_B,y),y)=x_Bu(x_B,y)f_1(x_B)f_2(y) - \int_{x_M}^{x_N} \triangle(x,y)u(x,y)dx$.
Pick a random point $P$ on $ND$, draw a horizontal line across $P$ that intersects curve $MLQ$ at $J$,  intersects segment $MG$ at $R$.
In region $CSMLID$ point $(x,y)$ gets item 2 deterministically and we assume it gets item 1 with probability $q_1(x)$.
\begin{eqnarray}
& &R(u(x_B,y_J),y_J)\nonumber\\
&=&u(x_B,y_J)x_Bf_1(x_B)f_2(y_J) - \int_{x_M}^{x_J}u(x,y_J)\triangle(x,y_J)dx -
\int_{x_J}^{x_B}u(x,y_J)\triangle(x,y_J)dx\nonumber\\
&=&[u(R)+\int_{x_M}^{x_J}q_1(x)dx+x_B-x_J]x_Bf_1(x_B)f_2(y_J)-\int_{x_M}^{x_J}[u(R)+\int_{x_R}^x q_1(l)dl]\triangle(x,y_J)dx \nonumber\\
& &-\int_{x_J}^{x_B}[u(R)+\int_{x_M}^{x_J}q_1(l)dl+x-x_J]\triangle(x,y_J)dx\nonumber
\end{eqnarray}

Because $u(x_B,y_J)=u(R)+\int_{x_M}^{x_J}q_1(x)dx+x_B-x_J=u(R)+x_B-\int_{x_M}^{x_J}(1-q_1(x))dx$, we have
$\frac{\partial u}{\partial x_J}(x_B,y_J) \leq 0$

\begin{eqnarray}
& &\frac{\partial R}{\partial x_J}(u(x_B,y_J),y_J)\nonumber\\
&=&(q_1(x_J)-1)x_Bf_1(x_B)f_2(y_J)-\int_{x_J}^{x_B}(q_1(x_J)-1)\triangle(x,y_J)dx\nonumber\\
&=&(q_1(x_J)-1)f_2(y_J)[x_Bf_1(x_B)-\int_{x_J}^{x_B}f_1(x)(3+PR(f_1(x))+PR(f_2(y_J)))dx]\nonumber\\
&=&(q_1(x_J)-1)f_2(y_J)v(x_J)\nonumber
\end{eqnarray}
In the last equality, we reset $v(x_J)=x_Bf_1(x_B)-\int_{x_J}^{x_B}f_1(x)(3+PR(f_1(x))+PR(f_2(y_J)))dx$.
Since $q_1(x_J)-1\leq 0$, $sgn(v(x_J))=-sgn(\frac{\partial R}{\partial x_J})$. Thus
$$sgn(v(x_J))=-sgn(\frac{\partial R}{\partial u}\frac{\partial u}{\partial x_J})=sgn(\frac{\partial R}{\partial u})$$

Since $3+PR(f_1(x))+PR(f_2(y))\geq 0$, $v(x_J)$ is weakly monotone in $x_J$.
By (3) and (4) of Lemma~\ref{lemma:pre}, while $P$ moves up, the intersection $J$ moves towards right, i.e., $x_J$ weakly increases. Then $v(x_J)$ weakly increases. So if $v(x_J)$ switches sign, it must switch sign from negative to positive. So $R_u(u(x_B,y_J),y_J)$ can only switch sign from negative to positive.

There are three cases for the sign of $R_u(u(x_B,y_J),y_J)$. It's similar to the proof in Lemma \ref{lemma:fourcanon}. WLOG, we can assume that
$v(x_J)=0$. Thus $v(x)\geq 0, x\in[x_J,x_D]$ and $v(x)\leq 0, x\in[x_A,x_J]$. Therefore $R_u(u(x_B,y_J),y_J)\geq 0, y\in[y_P,y_D]$, $R_u(u(x_B,y_J),y_J)\leq 0, y\in[y_N,y_P]$.
So $R_{MNDG}$ is increasing as $u(NP)$ decrease and is increasing as $u(PD)$ increase. With fixed $u(N)$, $u(P)$, and $u(D)$, by convexity, optimal $u(ND)$ comprises of two lines: the straight line (with extended line) across points $(y_P,u(P))(y_D,u(D))$, and straight line across point $(y_N,u(N))$ with slope $q_2(N)$.
\end{proof}

\begin{proof}\textbf{of Theorem~\ref{thm:4items}}
\begin{figure}
  \centering
  \includegraphics[width=6cm]{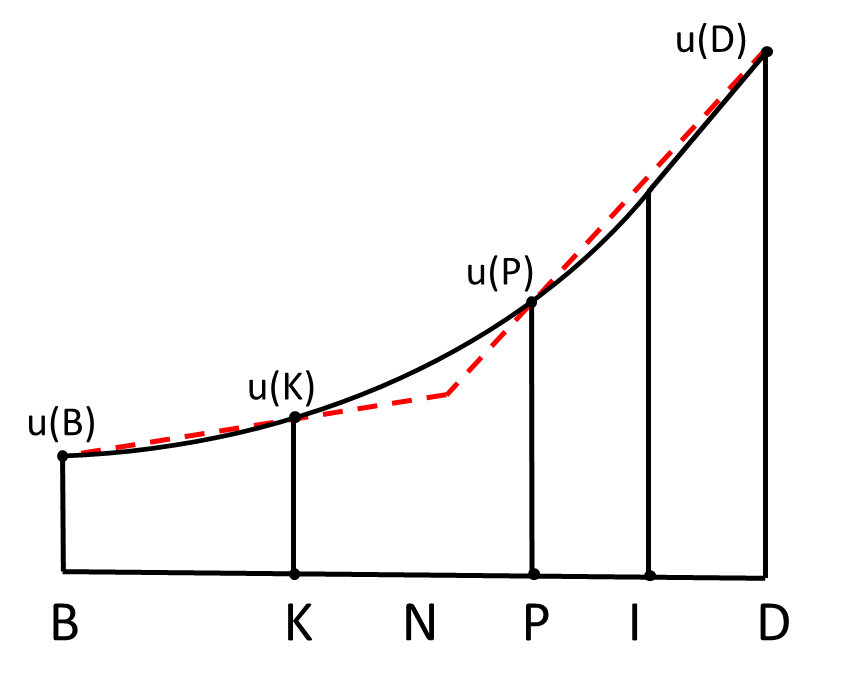}\\
  \caption{Utility function on BD.}
  \label{fig:example3}
\end{figure}

We sum up the conclusions drawn on different segments of $BD$ and settle final shape of $u(BD)$, subject to the convexity constraints and fixed $u(B),u(K),u(P)$, and $u(D)$.
In Fig.~\ref{fig:example3}, the black solid line is an arbitrary convex utility function. By Lemma \ref{lemma:fourcanon} and \ref{lemma:four2}, the optimal $u(BD)$ consists of at most 3 pieces:
the straight line (with extended line) across points $(y_B,u(B))(y_K,u(K))$,
the straight line (with extended line) across points $(y_P,u(P))(y_D,u(D))$, and the straight line across point $(y_N,u(N))$ with slope $q_2(N)$. In fact, we can prove an even stronger result: it turns out that the line across point $(y_N,u(N))$ is unnecessary. The remainder of the proof is to confirm this claim.

In Fig.~\ref{fig:example3}, the red dashed utility consists of two parts: straight line across points $(y_B,u(B))$ and $(y_K,u(K))$, and straight line across points $(y_P,u(P))$ and $(y_D,u(D))$.
We denote "original" utility function to be the black arbitrary convex utility function. We denote "new" utility function to be the red dashed utility function. We prove that the revenue based on new utility function is greater than or equal to the revenue based on the black original utility.
Therefore when $u(B),u(K),u(P),u(D)$ and their coordinates on the $y$-dimension $y_B,y_K,y_P,y_D$ are fixed,
the optimal utility function must be of the shape portrayed as the red dashed line.

First, we study what the graph representation looks like.
Since $u(Q),u(J),u(M),u(W)$ remain the same, they are still on the boundaries (see Fig.~\ref{fig:newproof}):
\begin{itemize}
\item $M$ is on the boundary between $(*,1)$ region and $(0,0)$ region.
\item $W$ is on the boundary between $(1,*)$ region and $(0,0)$ region.
\item $J$ is on the boundary between $(1,*)$ region and $(*,1)$ region.
\item $Q$ is on the boundary between $(*,1)$ region and $(1,*)$ region.
\end{itemize}
Let $M'$ denote the new intersection of the three parts: $(0,0),(*,1)$, and $(1,*)$.
The new boundary between $(1,*)$ and $(*,1)$ is the dashed line $QJM'$. The new boundary between $(*,1)$ and $(0,0)$ is the dashed line $SMM'$.
The new boundary between $(1,*)$ and $(0,0)$ is the dashed line $M'WE$.
Since $u(PD)$ and $u(KB)$ weakly increase, and $u(KP)$ weakly decreases, $M'$ must lie in $MWKN$ region, dashed $QJ$ lies in $GMLQ$ region, dashed $WE$ lies in $SMWE$ region. Draw horizontal line through $M'$ and it intersects $BD$ at $N'$, intersects the extended line of $GM$ at $O$.

With fixed $u(CD)$, given the boundary of $(1,*)$, we can calculate $u(BD)$ as follows: say $(x,y)$ is on the boundary between $(1,*)$ amd $(*,1)$, then $u(x_B,y)=u(x,y)+x_B-x=u(x,y_C)+y-y_C+x_B-x$. So we can define $u(BD)$ according to their boundary.
Let the original utility and revenue function based on the boundary $QLJMWE$ is $u^1$ and $R^1$.
Let the new utility and revenue function based on the boundary $QJM'W$ and dashed $WE$ is $u^3$ and $R^3$.
Our goal is to prove $R^1\leq R^3$.

$$R^1_{DGMWEB}\leq R^3_{DGM'WEB}$$
Since $R^1_{DGMWEB}=R^1_{DGMWK}+R^1_{WEBK}$(here $WE$ is the solid line) and $R^3_{DGM'WEB}=R^3_{DGMM'WK}+R^3_{WEBK}$(here, $WE$ is the dashed line).  It suffices to prove $R^1_{DGMWK}\leq R^3_{DGMM'WK}$ and $R^1_{WEBK}\leq R^3_{WEBK}$ separately.
\begin{eqnarray}
&&R^1_{WEBK}\nonumber\\
&=&\int_{WK}\mathbf{T}^1\cdot \mathbf{\hat n}ds+\int_{EB}\mathbf{T}^1\cdot \mathbf{\hat n}ds+\int_{y_B}^{y_K}u^1(x_B,y)x_Bf_1(x_B)f_2(y)dy\nonumber\\
&&-\int_{y_B}^{y_K}\int_{x_B-u(x_B,y)}^{x_B}u^1(x,y)\triangle(x,y)dxdy\nonumber\\
&=&\int_{WK}\mathbf{T}^3\cdot \mathbf{\hat n}ds+\int_{EB}\mathbf{T}^3\cdot \mathbf{\hat n}ds+\int_{y_B}^{y_K}R(u^1(x_B,y),y)dy\nonumber\\
&\leq&\int_{WK}\mathbf{T}^3\cdot \mathbf{\hat n}ds+\int_{EB}\mathbf{T}^3\cdot \mathbf{\hat n}ds+\int_{y_B}^{y_K}R(u^3(x_B,y),y)dy\nonumber\\
&=&R_{WEBK}^3\nonumber
\end{eqnarray}
Here $\mathbf{T}^i=zu^if(z)$ and we use the same $R(u(x_B,y),y)$ as that in Lemma~\ref{lemma:fourcanon}. What still remains to show is $R^1_{DGMWK}\leq R^3_{DGMM'WK}$.

In order to prove this claim, we introduce an intermediate utility and revenue as follows.
Let the intermediate utility and revenue function based on the boundary $QLJMM'W$ and dashed $WE$ is $u^2$ and $R^2$.
Consider the intermediate case, $MM'$ is the boundary shared by all three regions $(1,*),(*,1)$ and $(0,0)$. For any point $(x,y)$ on $MM'$,
we can calculate the utility of corresponding point $(x_B,y)$ as $u(x_B,y)=u(x,y)+x_B-x=u(x,y_C)+y-y_C+x_B-x$. In fact, according to Lemma~\ref{lemma:pre}, this case cannot happen and $u^2(BD)$ does not retain convexity any more. It is important to note that, here, we are only concerned with $R^2$, not the feasibility of $u^2$.
That is, we use $R^2$ as a number to facilitate the comparison between $R^1$ and $R^3$.

Now we show $R^1_{DGMWK}\leq R^3_{DGMM'WK}$.
\begin{eqnarray}
&&R^1_{DGMWK}\nonumber\\
&=&R^1_{DGMN}+R^1_{MWKN}\nonumber\\
&\leq&R^1_{DGMN}+R^2_{MM'WKN}\nonumber\\
&=&R^2_{DGMN}+R^2_{MM'N'N}+R^2_{M'WKN'}\nonumber\\
&=&R^2_{DGMM'N'}+R^2_{M'WKN'}\nonumber\\
&=&R^2_{DGMM'N'}+R^3_{M'WKN'}\nonumber\\
&\leq&R^3_{DGMM'N'}+R^3_{M'WKN'}\nonumber\\
&=&R^3_{DGMM'WK}\nonumber
\end{eqnarray}

The first inequality follows from a similar proof of that in $R^1_{WEBK}\leq R_{WEBK}^3$.

In order to show the second inequality, introduce $u^{22}$ and $u^{33}$ as follows:
\begin{displaymath}
u^{22}(x,y)=\left\{ \begin{array}{ll}
u^2(x,y) & (x,y)\in DGMM'N'\\
u^2(x,y_C)+y-y_C & (x,y)\in MM'O
\end{array}\right.
\end{displaymath}
\begin{displaymath}
u^{33}(x,y)=\left\{ \begin{array}{ll}
u^3(x,y) & (x,y)\in DGMM'N'\\
u^3(x,y_C)+y-y_C & (x,y)\in MM'O
\end{array}\right.
\end{displaymath}
\begin{eqnarray}
&&R^2_{DGMM'N'}-R^3_{DGMM'N'}\nonumber\\
&=&\int_{DGMM'N'}\mathbf{T}^2\cdot \mathbf{\hat n}ds +\int_{N'D}\mathbf{T}^2\cdot \mathbf{\hat n}ds
-\int_{y_{N'}}^{y_D}\int_{x_M}^{x_{N'}}\triangle(x,y)u^{22}(x,y)dxdy\nonumber\\
&&+\int_{MM'O}\triangle(z)u^{22}(z)dz\nonumber-R^3_{DGMM'N'}\nonumber\\
&=&\int_{y_{N'}}^{y_D}[x_Bu^{22}(x_B,y)f_1(x_B)f_2(y)-\int_{x_M}^{x_N}\triangle(x,y)u^{22}(x,y)dx ]dy\nonumber\\
&&-\int_{y_{N'}}^{y_D}[x_Bu^{33}(x_B,y)f_1(x_B)f_2(y)-\int_{x_M}^{x_N}\triangle(x,y)u^{33}(x,y)dx ]dy\nonumber\\
&=&\int_{y_{N'}}^{y_D}[R(u^{22}(x_B,y),y)-R(u^{33}(x_B,y),y)]dy\label{eq4}
\end{eqnarray}
In the last equality, we use the same $R(u(x_B,y),y)$ as in Lemma~\ref{lemma:four2}.

In the following, we show Equation (\ref{eq4}) is non-positive.
According to Lemma~\ref{lemma:four2},
$$R_u(u(x_B,y),y)\geq 0 ~~y\in[y_J,y_Q],u(x_B,y)\geq u(J)$$
$$R_u(u(x_B,y),y)\geq 0 ~~y\in[y_{M'},y_J],u(x_B,y)\leq u(J)$$
Because $u^{22}(x_B,y)\leq u^{33}(x_B,y) ~~y\in[y_P,y_D]$ and $u^{22}(x_B,y)\geq u^{33}(x_B,y) ~~y\in[y_{N'},y_P]$, we have $R(u^{22}(x_B,y),y)\leq R(u^{33}(x_B,y),y),~~y\in[y_{N'},y_P]$. So Equation $(\ref{eq4})\leq 0$, then $R^2_{DGMM'N'}\leq R^3_{DGMM'N'}$.

Up to now, we have proved that red dashed utility function on $BD$ segment indeed yield the highest revenue subject to fixed $u(B),u(K),u(P)$, and $u(D)$ (Fig.~\ref{fig:example3}). In other words, $u(BD)$ is piecewise linear with two pieces. Same for $u(CD)$.

We have showed region $CSEBD$ consists of 4 menu items. In other words, the whole mechanism consists of 5 menu items.
Say, points on $BD$ segments choose menus $(1,\alpha,t_\alpha),(1,\gamma,t_\gamma), ~\alpha\leq\gamma$. Points on $CD$ choose menus $(\beta,1,t_\beta),(\theta,1,t_\theta), ~\beta\leq\theta$.
What remains to show is that the top right two regions both allocate with probabilities (1,1) thus are in fact a unique region.

\citet[Theorem 16]{Manelli2007}  state that in the optimal mechanism, there must exist a segment chooses allocation $(1,1)$. By Lemma ~\ref{lemma:pre}, the $(1,1)$ region is on the top right corner of the rectangle. So in the optimal mechanism there are slopes of utility lines equal to 1 in both $BD$ and $CD$. Hence, $\gamma=\theta=1$. $BD$ and $CD$ share the same menu item $(1,1)$.

To sum up, the optimal mechanism menu consists of at most 4 menu-items:
\begin{center}
\begin{tabular}{|c|c|c|}
\hline
$q_1$ & $q_2$ & $t$\\
\hline
0 & 0 & 0\\
\hline
1 & $\alpha$ & $t_\alpha$\\
\hline
$\beta$ & 1 & $t_\beta$\\
\hline
1 & 1 & $t_1$\\
\hline
\end{tabular}
\end{center}
\end{proof}

\begin{proof}\textbf{of Theorem~\ref{theorem:6}}

\begin{figure}[H]
  \centering
  \includegraphics[width=6cm]{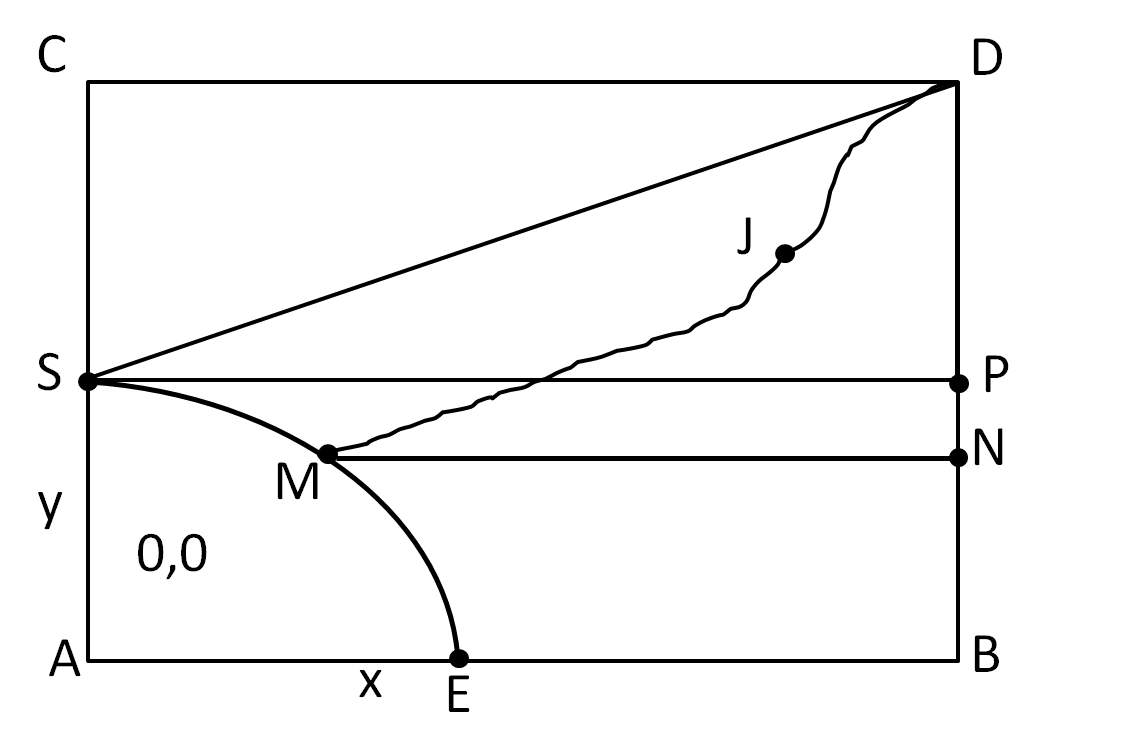}\\
  \caption{Optimal mechanism with 3 menu items.}
  \label{fig:bundle_extrem01}
\end{figure}

We consider two cases:

Case 1: When there is a region that chooses first item randomly and choose second item with probability 1, i.e., there is region with allocation
Look at Fig.~\ref{fig:bundle_extrem01}, $SMJDC$ represents the $(*,1)$ region.
When $y\in[y_B,y_N]$, we have $u(x_B,y)\leq x_B-x_A$, Lemma~\ref{lemma:fourcanon} implies
$$R_u(u(x_B,y),y)\geq f_2(y)[x_Bf_1(x_B)-\int_{x_A}^{x_B}f_1(x)[3+PR(f_1(x))+PR(f_2(y))]dx].$$

When $y\in[y_N,y_D]$, we have $x_J\geq x_A$, Lemma \ref{lemma:four2} implies
$$v(x_J)\geq  x_Bf_1(x_B)-\int_{x_A}^{x_B}f_1(x)[3+PR(f_1(x))+PR(f_2(y))]dx.$$

We have
\begin{eqnarray}
&&x_Bf_1(x_B)-\int_{x_A}^{x_B}f_1(x)[3+PR(f_1(x))+PR(f_2(y))]dx\nonumber\\
&=&x_Af_1(x_A)+\int_{x_A}^{x_B}[xf'_1(x)+f_1(x)]dx-\int_{x_A}^{x_B}f_1(x)[3+PR(f_1(x))+PR(f_2(y))]dx\nonumber\\
&=&x_Af_1(x_A)\int_{x_A}^{x_B}f_1(x)dx-\int_{x_A}^{x_B}f_1(x)[2+PR(f_2(y))]dx\nonumber\\
&=&\int_{x_A}^{x_B}f_1(x)[x_Af_1(x_A)-2-PR(f_2(y))]dx\geq0\nonumber
\end{eqnarray}
Thus $R_u(u(x_B,y),y)\geq 0$, $y\in[y_B,y_N]$ and $v(x_J)\geq 0$, $py\in[y_N,y_D]$.
Because $sgn(R_u(u(x_B,y),y))=sgn(v_J)\geq 0, y\in[y_N,y_D]$, $R_u(u(x_B,y),y)\geq 0, y\in[y_N,y_D]$.
By Lemma~\ref{lemma:fourcanon} and \ref{lemma:four2}, with fixed $u(B),u(N)$, and $u(D)$, optimal $u(ND)$ and $u(BN)$ are straight lines.
Using similar method as shown in the Theorem~\ref{thm:4items}, we can prove that with fixed $u(B)$ and $u(D)$, the optimal utility function is $u(x_B,y)=u(B)+q_{23}(y-y_B)~~y\in[y_B,y_D]$, where $q_{23}=\frac{u(D)-u(B)}{y_D-y_B}$. But we should notice that, this operation may change $u(AC)$. This leads to two further cases: (1) the operation changes $u(AC)$; (2) otherwise.  It can be seen case (1) is strictly more general, so we only need to consider case (1). 

Draw a vertical line through $S$ across $BD$ at $P$, we fix $u(B)$ and $u(D)$. Instead of letting $u(BD)$ be a straight line, we let $u(BD)$ be piece linear with one break point $u(P)$ in order not to change $u(AC)$. Let $u(P)=x_B-x_A$,
$$u(x_B,y)=u(P)+\frac{y-y_P}{y_D-y_P}\cdot (u(D)-u(P))\qquad y\in[y_P,y_D]$$
$$u(x_B,y)=u(B)+\frac{y-y_B}{y_P-y_B}\cdot (u(P)-u(B))\qquad y\in[y_B,y_P]$$
After this operation, points $u(S),u(P),u(D)$ are in a new plane. So the boundary of $(*,1)$ region, curve $MJD$, moves into $CSD$ region.
Using similar method as is shown in the Theorem~\ref{thm:4items}, total revenue will weakly increase.


Then we set
$$u(x,y_C)=u(C)+\frac{x-x_C}{x_D-x_C}\cdot(u(D)-u(C))\qquad x\in[x_C,x_D]$$
After this operation, points $u(C),u(S),u(D)$ are in a new plane. The boundary of $(*,1)$ region, originally curve $MJD$, becomes the straight line $SD$. So this operation will not affect the utilities on segment $AB$. Using same argument above, total revenue will weakly increase. Now there are 2 non-zero menu items chosen in region $CSPD$ and 1 non-zero menu item chosen in region $SMEBP$.
\citet[Theorem 16]{Manelli2007}  says in the optimal mechanism, there must exist a segment chooses $(1,1)$ allocation.
So these 2 non-zero menus chosen by region $CSPD$ must be the same menu $(1,1,t)$. So there are two non-zero menus in total:$(1,1,t)$ and $(1,\alpha, t_1)$.

Case 2: When there is no $(*,1)$ region.

Now Condition $"PR(f_1(x))\leq y_Af_2(y_A)-2,\forall x"$ is not enough, because when we do the same manipulation for $u(BD)$, $u(CS)$ will change at the same time. If a point have chosen non-zero menu item, it must choose first item with probability 1. According to this, we have the following equation:

\begin{eqnarray}
R_{SPDC}&=&\int_{SP}\mathbf{T}\cdot \mathbf{\hat n}ds+\int_{CD}\mathbf{T}\cdot \mathbf{\hat n}ds+\int_{PD}\mathbf{T}\cdot \mathbf{\hat n}ds+\int_{SC}\mathbf{T}\cdot \mathbf{\hat n}ds-\int_{SPDC}\triangle(z)u(z)dz\nonumber\\
&=&c+\int_{y_S}^{y_C}u(x_B,y)x_Bf_1(x_B)f_2(y)dy\nonumber\\
&&-\int_{y_S}^{y_C}(u(x_B,y)+x_A-x_B)x_Af_1(x_A)f_2(y)dy\nonumber\\
&&-\int_{y_S}^{y_C}\int_{x_A}^{x_B}\triangle(x,y)(u(x_B,y)+x-x_B)dxdy\nonumber\\
&=&c+\int_{y_S}^{y_C}R(u(x_B,y),y)dy\nonumber
\end{eqnarray}

Here, $c=\int_{SP}\mathbf{T}\cdot \mathbf{\hat n}ds+\int_{CD}\mathbf{T}\cdot \mathbf{\hat n}ds$. When we fix $u(D)$ and $u(P)$, $c$ is constant.
Let $R=u(x_B,y)x_Bf_1(x_B)f_2(y)-(u(x_B,y)+x_A-x_B)x_Af_1(x_A)f_2(y)-\int_{x_A}^{x_B}\triangle(x,y)(u(x_B,y)+x-x_B)dx$ and $R_u(u(x_B,y),y)$ denote the partial derivative wrt. $u(x_B,y)$.
\begin{eqnarray}
&&R_u(u(x_B,y),y)\nonumber\\
&=&[x_Bf_1(x_B)-x_Af_1(x_A)-\int_{x_A}^{x_B}f_1(x)[3+PR(f_1(x))+PR(f_2(y))]dx]\cdot f_2(y)\nonumber\\
&=&[-2-PR(f_2(y))]\cdot f_2(y)\leq0\nonumber
\end{eqnarray}
This means $R_{SPDC}$ is increasing as utilities on segment $PD$ decreases and is decreasing as utilities on segment $PD$ increases. Using same argument as in Case 1, we know total revenue is increasing as utilities on segment $PB$ increases and is decreasing as utilities on segment $PD$ decreases.
\begin{figure}[H]
  \centering
  \includegraphics[width=6cm]{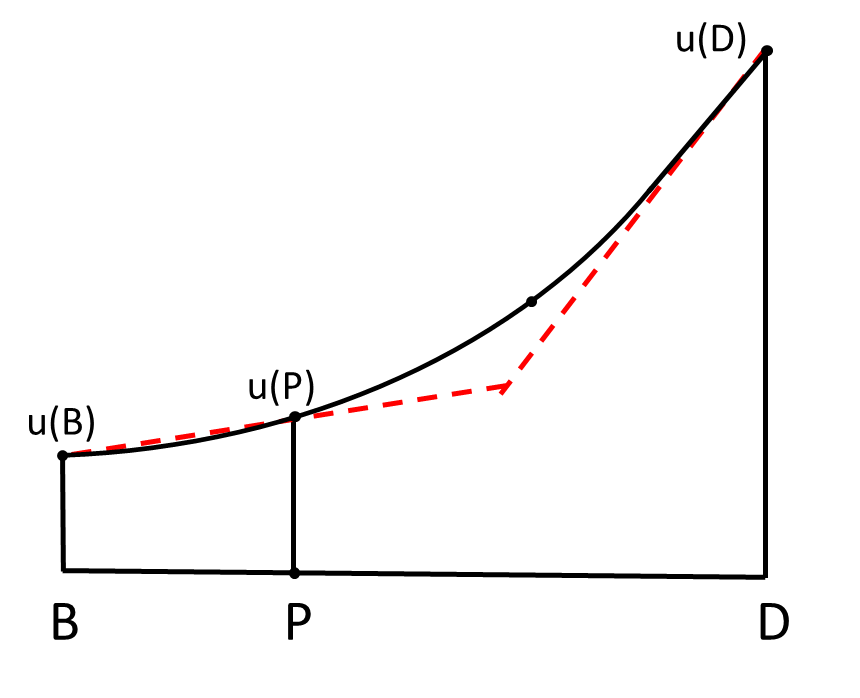}\\
  \caption{Utility function on BD, the red dashed line including point u(D) has derivative 1.}
  \label{fig:bundle_extrem02}
\end{figure}
As shown in Fig.~\ref{fig:bundle_extrem02}, the new mechanism with red dashed line $u(BD)$ yields higher revenue than the original mechanism with black solid line $u(BD)$. So there are two non-zero menu items in total:$(1,1,t)$ and $(1,\alpha, t_1)$.
$\Box$
\end{proof}

\begin{proof}\textbf{of Lemma~\ref{lemma:monotonePR2}}
Since point $M$ is on $(0,0)$ part, so $u(M)=0$.
For point $(x,y)$ in region $NMD$, $u(x,y)=u(M)+y-y_M+\int_{x_M}^xq_1(l)dl=y-y_M+\int_{x_M}^xq_1(l)dl$.

Rewrite the revenue formula for region $NMGD$ as follows,
\begin{eqnarray}
R_{NMGD}&=&\oint_{NMGD}\mathbf{T}\cdot \mathbf{\hat n}ds-\int_{NMGD}\triangle(z)u(z)dz\nonumber\\
&=&2\int_{x_N}^{x_D}y_Cu(x,y_C)f_1(y_C)f_1(x)dx-2\int_{NM}\mathbf{T}\cdot \mathbf{\hat n}ds\nonumber\\
& &-2\int_{x_N}^{x_D}\int_x^{y_C}u(x,y)\triangle(x,y)dydx\nonumber\\
&=&2\int_{x_N}^{x_D}x_B[y_C-y_M+\int_{x_N}^xq_1(l)dl]f_1(x_B)f_1(x)dx\nonumber\\
& &-2\int_{x_N}^{x_D}\int_x^{y_C}[y-y_M+\int_{x_N}^xq_1(l)dl]\triangle(x,y)dydx-2\int_{NM}\mathbf{T}\cdot \mathbf{\hat n}ds\nonumber\\
&=&2\int_{x_N}^{x_D}q_1(x)(\int_x^{x_D}[x_Bf_1(x_B)f_1(l)-\int_l^{x_D}\triangle(l,y)dy]dl)dx+C_1\nonumber\\
&=&2\int_{x_N}^{x_D}q_1(x)v(x)dx+C_1\nonumber
\end{eqnarray}

The rest terms are denoted by $C_1$ which only depends on $y_M$. When $y_M$ is fixed, $C_1$ is a constant.
In the last equality, let $v(x)=\int_x^{x_D}[x_Bf_1(x_B)f_1(l)-\int_l^{x_D}\triangle(l,y)dy]dl$, thus it is independent from $q_1$.

To find out the optimal $q_1$, divide $[x_N,x_D]$ into several intervals according to $v(x)$.
Since $v(x)$ is a continuous function, we can assume in intervals $v(x)\leq 0,x\in[n_i,d_i], \forall 1\leq i\leq l$, and $v(x)\geq 0,x\in[d_i,n_{i+1}],\forall 1\leq i\leq l-1,\textrm{and }m_1=x_N,b_n=x_D$. There is always a rational number in any interval, so the number of intervals is countable. It is possible that $d_1=x_N$, $n_l=x_D$ and $l$ may be infinite.
\begin{eqnarray}
&&\int_{x_N}^{x_D}q_1(x)v(x)dx\nonumber\\
&=&\sum_{i=1}^l\int_{n_i}^{d_i}q_1(x)v(x)dx+\sum_{i=1}^{l-1}\int_{d_i}^{n_{i+1}}q_1(x)v(x)dx\nonumber\\
&\leq&\sum_{i=1}^lq_1(n_i)\int_{n_i}^{d_i}v(x)dx+\sum_{i=1}^{l-1}q_1(n_{i+1})\int_{d_i}^{n_{i+1}}v(x)dx\label{eq3}\\
&=&q_1(n_1)\int_{x_N}^{x_D}v(x)dx+\sum_{i=2}^l(q_1(n_i)-q_1(n_{i-1}))\int_{d_{i-1}}^{d_l}v(x)dx\nonumber\\
&\leq&q_1(n_1)\int_{x_N}^{x_D}v(x)dx+\sum_{i=2}^l(q_1(n_i)-q_1(n_{i-1}))\int_{x'}^{x_D}v(x)dx\label{eq2}\\
&\leq&q_1(x_N)\int_{x_N}^{x'}v(x)dx+1\times\int_{x'}^{x_D}v(x)dx\nonumber
\end{eqnarray}

Let $x'$ denote the point in $[x_N,x_D]$ where $g(s)=\int_s^{x_D}v(x)dx$ achieves the maximum value in (\ref{eq2}).
The inequality~(\ref{eq3}) are only based on the facts that $q_1(x)$ is a weakly monotone function and they become equalities by setting $q_1(x)=q_1(x_N),x\in[x_N,x'), q_1(x)=1,x\in[x',x_D]$. So there are at most two pieces in $u(BD)$.
\end{proof}

\begin{proof}\textbf{of Lemma~\ref{lemma:monotonePR3}}
The proof is similar to Lemma~\ref{lemma:fourcanon}.
Rewrite the revenue formula of region $CSMN$ as follows,
\begin{eqnarray}
R_{CSMN}&=&\oint_{CSMN} \mathbf{T}\cdot \mathbf{\hat n}ds - \int_{CSMN}\triangle(z)u(z)dz\nonumber\\
&=&\int_{CSMN} \mathbf{T}\cdot \mathbf{\hat n}ds+\int_{x_C}^{x_N}y_Cu(x,y_C)f_1(y_C)f_1(x)dx\nonumber\\
& &-\int_{x_C}^{x_N}\int_{y_C-u(x,y_C)}^{y_C}\triangle(x,y)u(x,y)dydx\nonumber\\
&=&\int_{x_C}^{x_N}[y_Cu(x,y_C)f_1(y_C)f_1(x)-\int_{y_C-u(x,y_C)}^{y_C}(u(x,y_C)-y_C+y)\triangle(x,y)dy]dx+C_2\nonumber\\
&=&\int_{x_C}^{x_N}R(u(x,y_C),x)dx+C_2\nonumber
\end{eqnarray}
We let $C_2=\int_{CSMN} \mathbf{T}\cdot \mathbf{\hat n}ds$, it only depends on $u(C)$ and $u(N)$.
Let $R(u(x,y_C),x)=y_Cu(x,y_C)f_1(y_C)f_1(x)-\int_{y_C-u(x,y_C)}^{y_C}(u(x,y_C)-y_C+y)\triangle(x,y)dx$.
Pick point $K$ in $CN$, we have,
$$\frac{\partial R}{\partial u}(u(x_K,y_C),x_K)=f_1(x_K)[y_Cf_1(y_C)-\int_{y_C-u(x_K,y_C)}^{y_C}f_1(y)[3+PR(f_1(x_K))+PR(f_1(y))]dy]$$
While $x_K$ increases, $u(x_K,y_C)$ and $PR(f(x_K))$ increase, $y_Cf_1(y_C)-\int_{y_C-u(x_K,y_C)}^{y_C}f_1(y)[3+PR(f_1(x_K))+PR(f_1(y))]dy$ decreases. WLOG, we can assume
$\frac{\partial R}{\partial u}(u(x,y_C),x)\geq 0, x\in [x_C,x_K]$ and $\frac{\partial R}{\partial u}(u(x,y_C),x)\leq 0, x\in [x_K,x_N]$.
Revenue is increasing as $u(KN)$ decrease. Revenue is increasing as $u(CK)$ increase.
Let $u^1(x,y_C)=u(C)+\frac{u(K)-u(C)}{x_K-x_C}(x-x_C)$, $u^2(x,y_C)=u(N)+q_1(x_N,y_N)(x-x_N)$, then $u(x,y_C)=max(u^1(x,y_C),u^2(x,y_C)), ~~x\in[x_C,x_N]$ gives the optimal revenue with fixed $u(C),u(K),u(N)$ and $q_1(x_N,y_C)$. So optimal $u(CN)$ comprises at most two pieces.
\end{proof}

\begin{proof}\textbf{of Theorem~\ref{theorem:2}}
\begin{figure}
  \centering
  \includegraphics[width=6cm]{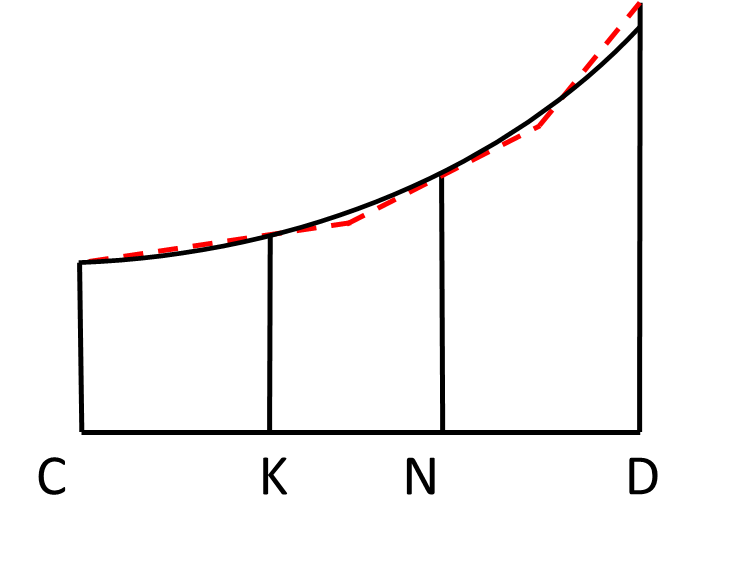}\\
  \caption{Symmetric value distribution.}
  \label{fig:example5}
\end{figure}

When $q_1(x_N)$ are fixed, $q_1(x) (x\in(x_N,x_D])$ is irrelevant to $R_{CSMN}$. When $u(C)$ and $u(N)$ are fixed, $u(x,y_B) (x\in(x_C,x_M))$ is irrelevant to $R_{NMGD}$.

Let's look at Fig.~\ref{fig:example5}. When $u(C),u(K),u(N),q_1(x_N)$ are fixed, the buyer utility on the red dashed line is greater than that on the black solid line between $CK$. The buyer utility on the red dashed line is less than that on the black solid line between $KN$. According to Lemma ~\ref{lemma:monotonePR2} and ~\ref{lemma:monotonePR3}, $R_{CSMN}$ and $R_{NMGD}$ are larger, so the total revenue is larger. Furthermore, optimal utility function on $ND$ is piecewise linear with two slopes: $q_1(x_N)$ and $1$. To sum up, the optimal utilities on $CD$ is piecewise linear with at most three pieces. Therefore, the optimal mechanism is of the following form:

\begin{center}
\begin{tabular}{|c|c|c|}
\hline
$q_1$ & $q_2$ & $t$\\
\hline
0 & 0 & 0\\
\hline
1 & $\alpha_1$ & $t_1$\\
\hline
1 & $\alpha_2$ & $t_2$\\
\hline
$\alpha_1$ & 1 & $t_1$\\
\hline
$\alpha_2$ & 1 & $t_2$\\
\hline
1 & 1 & $t_3$\\
\hline
\end{tabular}
\end{center}
\end{proof}

\begin{proof}\textbf{of Lemma~\ref{lemma:uni1}}
For $(x,y)$ in $CSEBD$, the utility of choosing $(q,1-q)$ is $xq+y(1-q)-t(q)=(x-y)q+y-t(q)$. The buyer will choose the menu item that achieves $max_{q}\{(x-y)q+y-t(q)\}=y+max_{q}\{(x-y)q-t(q)\}$. This means that which menu item will be chosen depends entirely on the value of $x-y$. For two points $(y_1+l,y_1)$ and $(y_2+l,y_2)$ in $CSEBD$, they must share the same allocation. Use $q_2(y)$ as a shortcut of $q_2(x_B,y)$.

Fixing $u(W)$ and $q_2(y_W)$, for any point $(x,y)$ in region $EBW$, the buyer's utility is
\begin{eqnarray}
u(x,y)&=&u(W)-(u(W)-u(x_B,y+x_B-x))-(u(x_B,y+x_B-x)-u(x,y))\nonumber\\
&=&u(W)-((x_B-x)(1-q_2(y+x_B-x))+(x_B-x)q_2(y+x_B-x))-\int_{x_B+y-x}^{y_W}q_2(l)dl\nonumber\\
&=&u(W)-(x_B-x)-\int_{x_B+y-x}^{y_W}q_2(l)dl\nonumber
\end{eqnarray}
Revenue formula in this region is,
\begin{eqnarray}
R_{EBW}&=&\oint_{EBW} \mathbf{T}\cdot \mathbf{\hat n}ds - \int_{EBW}\triangle(z)u(z)dz\nonumber\\
&=&\int_{WE}\mathbf{T}\cdot \mathbf{\hat n}ds+\int_{BW}\mathbf{T}\cdot \mathbf{\hat n}ds+\int_{EB}\mathbf{T}\cdot \mathbf{\hat n}ds- \int_{EBW}\triangle(z)u(z)dz\nonumber
\end{eqnarray}

The forms of the second and third terms in the expression above are as follows. Here, $a_i,b_i,i=1,...,7$ and $C_j,v_j,j=1,...,4$ are expressions independent of $q_2(y), y\in[y_B,y_W]$.

\begin{eqnarray}
& &\int_{BW}\mathbf{T}\cdot \mathbf{\hat n}ds=\int_{x_B}^{y_W}y_Bu(x_B,y)dy=\int_{a_1}^{b_1}h_1(y)[\int_{a_2(y)}^{b_2(y)}q_2(y)dx+g_1(y)]dy\nonumber\\
&=&\int_{y_B}^{y_W}q_2(y)v_1(y)+C_1\nonumber
\end{eqnarray}
\begin{eqnarray}
& &\int_{EB}\mathbf{T}\cdot \mathbf{\hat n}ds=\int_{x_E}^{x_B}-y_Au(x,y_A)dx=\int_{a_3}^{b_3}h_2(x)[\int_{a_4(x)}^{b_4(x)}q_2(y)dy+g_2(x)]dx\nonumber\\
&=&\int_{y_C}^{y_W}q_2(y)v_2(y)+C_2\nonumber
\end{eqnarray}

Since $f_1$ and $f_2$ are uniform distributions, $\triangle(x,y)=3f_1(x)f_2(y)+xf'_1(x)+yf'_2(y)=3f$ is a constant.
\begin{eqnarray}
&&\int_{EBW}\triangle(z)u(z)dz=\int_{y_C}^{y_W}\int_{x_B-y+y_C}^{x_B}3fu(x,y)dxdy\nonumber\\
&=&\int_{a_5}^{b_5}\int_{a_6(y)}^{b_6(y)}[\int_{a_7(x,y)}^{b_7(x,y)}q_2(l)dl+g_3(x,y)]dxdy=\int_{y_C}^{y_W}q_2(y)v_3(y)+C_3\nonumber
\end{eqnarray}
\begin{eqnarray}
\textrm{Thus,}&&R_{EBW}=\int_{WE}\mathbf{T}\cdot \mathbf{\hat n}ds+\int_{y_C}^{y_W}q_2(y)(v_1(y)+v_2(y)+v_3(y))dy+C_1+C_2+C_3\nonumber\\
&=&\int_{y_C}^{y_W}q_2(y)v_4(y)dy+C_4\nonumber
\end{eqnarray}

The following proof is similar to Lemma \ref{lemma:monotonePR2}, optimal $q_2(y) ~~y\in[y_C,y_W]$ comprises two parts.
There is $y'\in[y_C,y_W]$ such that

\begin{displaymath}
q_2(y)=\left\{\begin{array}{ll}
0 & y\in[y_C,y')\\
q_2(y_W) & y\in[y',y_W]
\end{array}\right.
\end{displaymath}
The utility function is
\begin{displaymath}
u(x_B,y)=\left\{\begin{array}{ll}
u(W)+(y-y_W)q_2(y_W) & y\in[y',y_W]\\
u(W)+(y'-y_W)q_2(y_W) & y\in[y_B,y']
\end{array}\right.
\end{displaymath}
Hence $u(x_B,y)\geq u(W)+(y'-y_W)q_2(y_W)\geq u(W)+(y_B-y_W)$.
Because $WE$ is a 45 degree line, $x_B-x_E=y_W-y_B$. Then $u(W)=(y_W-y_B)q_2(y_W)+(x_B-x_E)(1-q_2(y_W))=y_W-y_B$. We get $u(x_B,y)\geq 0$.
This new utility function satisfies the convexity and nonnegative property, so it's feasible. Therefore, $u(BW)$ comprises at most two pieces.
\end{proof}

\begin{proof}\textbf{of Lemma~\ref{lemma:uni2}}
\begin{eqnarray}
&&R_{GSMEWD}\nonumber\\
&=&\int_{GSMEW}\mathbf{T}\cdot \mathbf{\hat n}ds +\int_{GD}\mathbf{T}\cdot \mathbf{\hat n}ds +\int_{WD}\mathbf{T}\cdot \mathbf{\hat n}ds  - \int_{GSMD}3fu(z)dz-\int_{DMEW}3fu(z)dz\nonumber\\
&=&C+\int_{x_G}^{x_D}y_Cu(x,y_C)fdx-3f\int_{x_G}^{x_D}\int_{y_C-u(x,y_C)}^{y_C}(y-y_C+u(x,y_C))dydx\nonumber\\
& &+\int_{y_W}^{y_D}x_Bu(x_B,y)fdy-3f\int_{y_W}^{y_D}\int_{x_W-u(x_W,y)}^{x_W}(x-x_W+u(x_W,y))dxdy\nonumber\\
&=&C+\int_{x_G}^{x_D}fu(x,y_C)(y_C-\frac32u(x,y_C))dx+\int_{y_W}^{y_D}fu(x_B,y)(x_B-\frac32u(x_B,y))dy\nonumber
\end{eqnarray}
So, for fixed $u(D)$ and $u(W)$, when $u(x_B,y)> \frac13x_B, x\in(y_W,y_D)$, revenue is decreasing in $u(x_B,y)$; when $u(x_B,y)< \frac13x_B, y\in(y_W,y_D)$, revenue is increasing in $u(x_B,y)$. Similar to the proof in Lemma~\ref{lemma:fourcanon}, WLOG,
we can assume there is a point $P$ on $WD$ segment, such that $u(x_B,y_P)=\frac13x_B$. Since $u(PD)\geq u(P)$, revenue is increasing as $u(PD)$ decrease. Since $u(WP)\leq u(P)$, revenue is increasing as $u(WP)$ increase. Let $u^1(x_B,y)=u(W)+\frac{u(P)-u(W)}{y_P-y_W}(y-y_W)$, $u^2(x_B,y)=u(D)+q_2(y_D)(y-y_D)$. Then $u(x_B,y)=max(u^1(x_B,y),u^2(x_B,y))~~y\in[y_W,y_D]$ gives the optimal revenue with fixed $u(W),u(P),u(D)$, and $q_2(y_D)$. So optimal $u(WD)$ comprises at most two pieces.
\end{proof}

\begin{proof}\textbf{of Theorem \ref{thm:uni}}
\begin{figure}
  \centering
  \includegraphics[width=6cm]{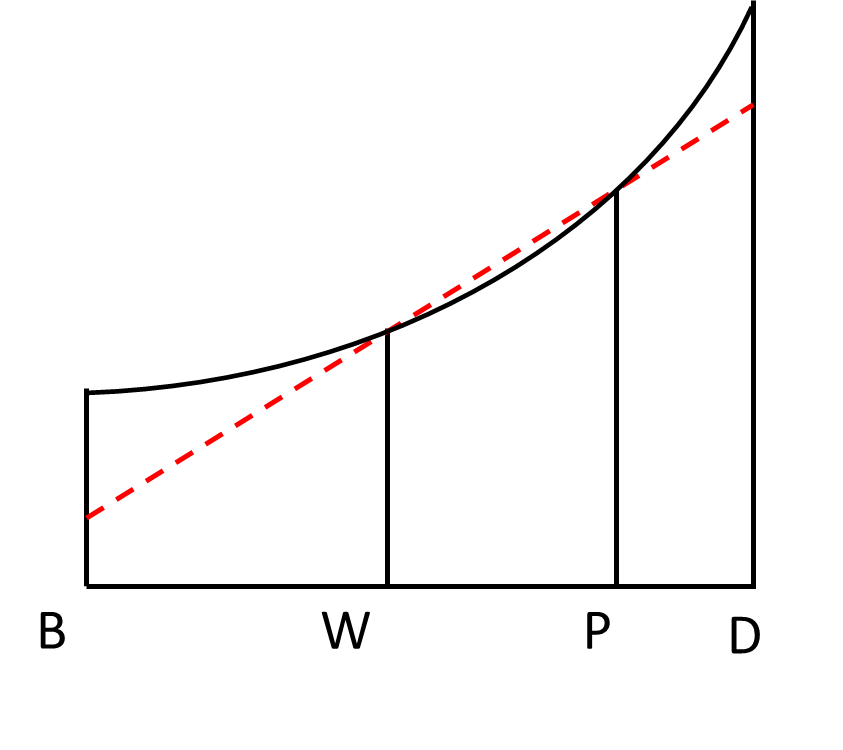}\\
  \caption{optimal utility function.}
  \label{fig:example7}
\end{figure}

We can now settle the optimal utilities on $BD$, subject to fixed values of $u(W)$ and $u(P)$ as well as the convexity of $u$.
Let $\alpha=\frac{u(P)-u(W)}{y_P-y_W}$ and $t_\alpha=(1-\alpha)x_P+\alpha y_P-u(P)$. Adding a new menu item $(1-\alpha, \alpha, t_\alpha)$,
the utility of choosing this new item is $u^\alpha$, which is denoted by the red dashed line shown in Fig.~\ref{fig:example7}.
Let $u^D$ be the utility obtained by choosing the same menu item as point $D$: $(1-q_2(y_C),q_2(y_C),t^D)$. Thus, for any point $(x,y_C)$ on $CD$, we have
\begin{eqnarray}
&&u(x,y_C)\geq u^D(x,y_C)=(1-q_2(y_C))x+q_2(y_C)y_C-t^D\nonumber\\
&&\geq(1-q_2(y_C))x_D+(1-q_2(y_C))(x-x_D)+q_2(y_C)-t^D\nonumber\\
&&=u(D)+(1-q_2(y_C))(x-x_D)\geq u^\alpha(D)+(1-\alpha)(x-x_D)\nonumber\\
&&=(1-\alpha)x_D+\alpha y_C-t^\alpha+(1-\alpha)(x-x_D)=u^\alpha(x,y_C)\nonumber
\end{eqnarray}

Thus, after adding this new menu item, $u(CD)$ does not change and $u(PW)$ will weakly increase.

In the optimal mechanism with fixed $q_2(y_W)$, according to Lemma \ref{lemma:uni1}, types on $BW$ choose only 2 menu items: $(1-\alpha, \alpha, t_\alpha)$ and $(1,0,t_1)$ for some $t_1$. Using same arguments as $CD$, for $BD$, we have another two menu items: $(1-\beta, \beta, t_\beta)$ and $(0,1,t_2)$ for some $t_2$.

We now consider a brand new menu with only the above four menu items and $(0,0,0)$. Compared to the old menu, $u(PD)$ weakly decreases while $u(PW)$ weakly increase. According to Lemma \ref{lemma:uni2}, $R_{GSMEWD}$ weakly increases. For fixed $q_2(y_W)$ and $u(W)$, $R_{EBW}$ is maximized.
For fixed $q_1(x_G,y_G)$ and $u(G)$, $R_{CSG}$ is maximized. The total revenue weakly increase. The optimal menu consists of at most 5 menu items:
\begin{center}
\begin{tabular}{|c|c|c|}
\hline
$q_1$ & $q_2$ & $t$\\
\hline
0 & 0 & 0\\
\hline
1 & 0 & $t_1$\\
\hline
$1-\alpha$ & $\alpha$ & $t_\alpha$\\
\hline
0 & 1 & $t_2$\\
\hline
$\beta$ & $1-\beta$ & $t_\beta$\\
\hline
\end{tabular}
\end{center}
\end{proof}

\end{appendix}

\end{document}